\documentclass[a4paper,UKenglish,cleveref, autoref, thm-restate]{lipics-v2021}

\usepackage{cmll}
\usepackage[all]{xy}
\usepackage{stmaryrd}
\nolinenumbers

\DeclareFontEncoding{LS1}{}{}
\DeclareFontSubstitution{LS1}{stix}{m}{n}
\DeclareSymbolFont{symbols2stix}      {LS1}{stixfrak} {m} {n}
\DeclareMathSymbol{\Lbrbrak}                  {\mathopen} {symbols2stix}{"22}
\DeclareMathSymbol{\Rbrbrak}                  {\mathclose}{symbols2stix}{"23}
\DeclareMathSymbol{\lblkbrbrak}               {\mathopen} {symbols2stix}{"36}
\DeclareMathSymbol{\rblkbrbrak}               {\mathclose}{symbols2stix}{"37}

\DeclareSymbolFont{symbolsstix}       {LS1}{stixscr}  {m} {n}
\DeclareSymbolFontAlphabet{\mathscr} {symbolsstix}

\newdir{C}{%
!/4.5pt/@{(}*:(1,-.2)@{}*:(1,+.2)@_{}}

\newdir{|>}{%
!/4.5pt/@{|}*:(1,-.2)@{>}*:(1,+.2)@_{>}}

\newcommand{\ttrue}{\mathbf{t\!t}}
\newcommand{\ffalse}{\mathbf{f\!f}}

\newcommand{\imc}{\rightarrowtriangle}

\newcommand{\iso}{\cong}

\newcommand{\qu}{\mathbf{q}}
\newcommand{\q}{{{\mathscr{q}}}}
\newcommand{\p}{{{\mathscr{p}}}}
\renewcommand{\r}{{{\mathscr{r}}}}

\newcommand{\mconflict}{\!\!\xymatrix@C=15pt{\, \ar@{~}[r]& \,}\!\!\!\!\!}
\newcommand{\intr}[1]{\llbracket #1 \rrbracket}

\newcommand{\pol}{\mathrm{pol}}
\newcommand{\just}{\mathsf{just}}

\newcommand{\tuple}[1]{\langle #1 \rangle}
\newcommand{\id}{\mathrm{id}}
\DeclareMathOperator{\dom}{dom}

\newcommand{\pred}{\mathsf{pred}}
\newcommand{\suc}{\mathsf{succ}}
\newcommand{\Aug}{\mathsf{Aug}}
\newcommand{\prefix}{\sqsubseteq}

\newcommand{\clone}{\mathrel{\approx}}
\newcommand{\up}{\mathop{\uparrow}}

\newcommand{\tq}{\;|\;}

\newcommand{\bool}{\mathbf{bool}}
\newcommand{\nat}{\mathbf{nat}}

\newcommand{\bq}{\mathbf{q}}

\newcommand{\ev}[1]{|#1|}
\newcommand{\conf}[1]{\mathscr{C}(#1)}
\newcommand{\N}{\mathbb{N}}

\newcommand{\tto}{\Rightarrow}
\newcommand{\Plays}{\mathsf{Plays}}
\newcommand{\pto}{\rightharpoonup}
\newcommand{\Inn}{\mathsf{Inn}}
\newcommand{\deseq}[1]{\Lbrbrak #1 \Rbrbrak}
\newcommand{\x}{\mathsf{x}}
\newcommand{\coll}[1]{\lblkbrbrak #1 \rblkbrbrak}
\newcommand{\Rel}{\mathsf{Rel}}
\renewcommand{\oc}{!}
\newcommand{\pos}[1]{\coll{#1}}
\newcommand{\positions}[1]{\lblkbrbrak\hspace{-2pt}| #1 |\hspace{-2pt}\rblkbrbrak}
\newcommand{\twin}{\mathsf{Fork}}
\newcommand{\caus}[1]{\hat{#1}}
\newcommand{\cod}{\mathsf{cod}}
\newcommand{\init}{\mathsf{init}}
\newcommand{\Trees}{\mathsf{Trees}}

\newcommand{\pview}[1]{\ulcorner\!#1\!\urcorner}

\newcommand{\pviews}[1]{\ulcorner\!\!\ulcorner #1 \urcorner\!\!\urcorner}

\newtheorem{quest}{Question}
\newtheorem{fact}{Fact}

\definecolor{Blue}{rgb}{.2,.2,.6}

\definecolor{Green}{cmyk}{1,0,1,0}
\newcommand{\green}[1]{{\color{Green}#1}}
\definecolor{Red}{rgb}{1,0,0}
\newcommand{\red}[1]{{\color{Red}#1}}
\definecolor{Gray}{rgb}{.7,.7,.7}
\newcommand{\gray}[1]{{\color{Gray}#1}}
\definecolor{Cyan}{cmyk}{1,0,0,0}
\newcommand{\cyan}[1]{{\color{Cyan}#1}}

\usepackage{todonotes}

\usepackage{tikz}
\newcommand*\circled[1]{\tikz[baseline=(char.base)]{
            \node[shape=circle,draw,inner sep=1pt] (char) {#1};}}

\title{Positional Injectivity for Innocent Strategies}

\author{Lison Blondeau-Patissier}{Univ Lyon, EnsL, UCBL, CNRS, LIP, F-69342, LYON Cedex 07, France}{Lison.Blondeau-Patissier@ens-lyon.fr}{}{}

\author{Pierre Clairambault}{Univ Lyon, EnsL, UCBL, CNRS, LIP, F-69342, LYON Cedex 07, France}{Lison.Blondeau-Patissier@ens-lyon.fr}{}{} 

\authorrunning{L. Blondeau-Patissier and P. Clairambault} 

\Copyright{Lison Blondeau-Patissier and Pierre Clairambault} 



\funding{Work supported by the ANR project DyVerSe
(ANR-19-CE48-0010-01); and by the Labex MiLyon (ANR-10-LABX-0070) of
Universit\'e de Lyon, within the program ``Investissements d'Avenir''
(ANR-11-IDEX-0007), operated by the French National Research Agency
(ANR).}
\acknowledgements{
We thank the reviewers, whose comments greatly helped improve
the paper.}

\hideLIPIcs  

\EventEditors{Naoki Kobayashi}
\EventNoEds{1}
\EventLongTitle{6th International Conference on Formal Structures for Computation and Deduction (FSCD 2021)}
\EventShortTitle{FSCD 2021}
\EventAcronym{FSCD}
\EventYear{2021}
\EventDate{July 17--24, 2021}
\EventLocation{Buenos Aires, Argentina (Virtual Conference)}
\EventLogo{}
\SeriesVolume{195}
\ArticleNo{17}

\begin{document}

\maketitle

\begin{abstract}
In asynchronous games, Melli\`es proved that innocent strategies are
\emph{positional}: their behaviour only depends on the position,
not the temporal order used to reach it. This insightful result
shaped our understanding of the link between dynamic
(\emph{i.e.} game) and static (\emph{i.e.} relational) semantics.

In this paper, we investigate the positionality of innocent strategies
in the traditional setting of Hyland-Ong-Nickau-Coquand pointer games.
We show that though innocent strategies are not positional, total finite
innocent strategies still enjoy a key consequence of positionality,
namely \emph{positional
injectivity}: they are entirely determined by their
positions. Unfortunately, this does not hold in general: we show a
counter-example if finiteness and totality are lifted. For finite
partial strategies we leave the problem open; we show however the
partial result that two strategies with the same positions must have the
same P-views of maximal length.
\end{abstract}

\section{Introduction}

\emph{Game semantics} presents higher-order computation interactively as
an exchange of tokens in a two-player game between Player (the
program under study), and Opponent (its execution
environment) \cite{HylandO00,AbramskyJM00}.  Game semantics has had a
strong theoretical impact on denotational semantics, achieving 
full abstraction results for languages for which other tools
struggle. 

At the heart of Hyland and Ong's celebrated model \cite{HylandO00} are
\emph{innocent strategies}, matching \emph{pure}
programs. They matter conceptually and technically: many full
abstraction results rely on innocent strategies and their
definability properties. Accordingly, innocence is perhaps
the most studied notion on the foundational side of game semantics, with
questions including categorical reconstructions
\cite{DBLP:conf/lics/HarmerHM07}, alternative definitions 
\cite{DBLP:journals/tcs/Mellies06,DBLP:journals/cuza/HirschowitzP12},
non-deterministic
\cite{DBLP:conf/lics/TsukadaO15,DBLP:conf/csl/CastellanCW14}, concurrent
\cite{CastellanCW15}, or quantitative
\cite{DBLP:journals/corr/TsukadaO14,DBLP:conf/lics/CastellanCPW18}
extensions. In particular, our modern understanding of
innocence is shaped 
by Melli\`es' homotopy-theoretic reformulation 
in asynchronous games \cite{DBLP:journals/tcs/Mellies06}. In this paper,
Melli\`es also introduced an important result: innocent strategies 
are \emph{positional}. 

\emph{Positionality} is an elementary notion on games on graphs: a
strategy is positional if its behaviour only depends on the current node
-- the ``position'' -- and not the path leading there. In standard game
semantics there is, at first sight, no clear notion of position: plays
are primitive, and it is not clear what is the ambient graph.
In contrast, asynchronous games and relatives (\emph{e.g.} concurrent
games) admit a transparent notion of position:
two plays reach the same
position if they feature the same moves, though not
necessarily in the same order. In investigating positionality,
Melli\`es' motivation was to bridge standard play-based game semantics
with 
more static, \emph{relational}-like semantics
\cite{DBLP:conf/lics/AbramskyM99,DBLP:journals/tcs/Ehrhard12}. Indeed, points of the \emph{web} in
relational semantics correspond to certain positions in game
semantics. Positionality of innocent strategies entails that they are
entirely defined by their positions (a property we shall call
\emph{positional injectivity}), so that collapsing game
to relational semantics corresponds exactly to keeping only
certain positions. See \cite{mall} for a recent
account. 

Now, traditional Hyland-Ong arena games are by no means
disconnected from those developments: bridges with relational semantics
were also investigated there, notably by Boudes
\cite{DBLP:conf/tlca/Boudes09}. There, points of
the web match so-called \emph{thick subtrees}, pomsets
representing partial explorations of the arena with duplications.
This provides
\emph{positions} for Hyland-Ong games.
But then, are innocent strategies still positional? Though it came to us as a
surprise, it is not hard to find a counter-example.
So we focus on the key weakening of the question: are
innocent strategies \emph{positionally injective}? Our main result is
positive, for \emph{total finite} innocent strategies. We first
link Hyland-Ong innocence with an alternative, causal
formulation inspired from concurrent games \cite{mall}, allowing a
transparent link between a strategy and its positions. Drawing
inspiration from the proof of injectivity of the relational model for
MELL proof nets \cite{Carvalho16}, we show how to track down
duplications in certain well-engineered positions to recover a
sufficient portion of the causal structure; and deduce positional
injectivity.  However, we show that in the general case (without
\emph{finiteness} and \emph{totality}), positional injectivity fails.
Finally, for finite (but not total) innocent strategies we show a
partial result, namely that two strategies with the same positions have
the same P-views of maximal length.  

Tsukada and Ong \cite{DBLP:conf/lics/TsukadaO16} show an injective
collapse from a category of innocent strategies onto the relational model.
Their collapse is similar to ours, with an important distinction: they
label moves in each play, coloring contiguous Opponent/Player
pairs identically. Labels survive the collapse, allowing to read back
causal links directly. This is possible because the web of atomic types
is set to comprise countably many such labels -- but then, the
correspondence between positions and points of the web is lost.  In
contrast, our theorem requires us to prove injectivity directly, without
such labeling.

In Section \ref{sec:inn_pos} we introduce the setting and state our main
result. In Section \ref{sec:caus} we reformulate the problem via a \emph{causal}
presentation of game semantics. In Section \ref{sec:pos_inj} we present
the proof of positional injectivity for total finite
innocent strategies. In Section \ref{sec:beyond_total_finite}, we show
some partial results beyond total finite strategies. Finally, in Section
\ref{sec:conclusion}, we conclude. 
Detailed proofs are attached in appendix.

\section{Innocent Strategies and Positions}
\label{sec:inn_pos}

\subsection{Arenas and Constructions}\label{sect:gs-arena}

We start this paper by giving a definition of \emph{arenas}, which
represent \emph{types}. 

\begin{definition}\label{def:arena}
An \textbf{arena} is $A = \tuple{\ev{A}, \leq_A, \lambda_A}$ where
$\tuple{\ev{A}, \leq_A}$ is a partial order, and $\lambda_A : \ev{A} \to
\{-, +\}$ is a \textbf{polarity function}. Moreover, these data must
satisfy:
\[
\begin{array}{rl}
\text{\emph{finitary:}} & \text{for all $a \in \ev{A}$, $[a]_A = \{a'
\in \ev{A} \mid a' \leq_A a\}$ is finite,}\\
\text{\emph{forestial:}} & \text{for all $a_1, a_2 \leq_A a$, then $a_1
\leq_A a_2$ or $a_2 \leq_A a_1$,}\\
\text{\emph{alternating:}} & \text{for all $a_1 \imc_A a_2$, then
$\lambda_A(a_1) \neq \lambda_A(a_2)$,}\\
\text{\emph{negative:}} & \text{for all $a \in \min(A) = \{a \in
\ev{A}\mid \text{$a$ minimal}\}$, $\lambda_A(a)
=-$,}
\end{array}
\]
where $a_1 \imc_A a_2$ means $a_1 <_A a_2$ with no event strictly in
between.
\end{definition}

Though our notations differ superficially, our arenas are similar to
\cite{HylandO00}. They present observable computational events (on a
given type) along with their causal dependencies: positive
moves are due to Player / the program, and negative moves to
Opponent / the environment. We show in Figures
\ref{fig:ar_bool} and \ref{fig:ar_nat}, read from top to bottom, the
representation of the datatypes $\bool$ and $\nat$ as
arenas. Opponent initiates the execution with $\bq^-$,
annotated so as to indicate its polarity, and Player may respond any
possible value, with a positive move.

\begin{figure}[t]
\begin{minipage}{.25\linewidth}
\[
        \xymatrix@C=0pt@R=8pt{
        & \bq^- \ar@{.}[dl] \ar@{.}[dr] & \\
        \ttrue^+ && \ffalse^+
        }
\]
\caption{Arena $\bool$}
\label{fig:ar_bool}
\end{minipage}
\hfill
\begin{minipage}{.25\linewidth}
\[
\xymatrix@C=0pt@R=10pt{
&& \bq^- \ar@{.}[dll] \ar@{.}[dl] \ar@{.}[d]
\ar@{.}[drr] && \\
0^+ & 1^+ & 2^+ && \ldots
}
\]
\caption{Arena $\nat$}
\label{fig:ar_nat}
\end{minipage}
\hfill
\begin{minipage}{.45\linewidth}
\[
\xymatrix@C=10pt@R=-3pt{
(o	\ar@{}[r]|\to&o)\ar@{}[r]|\to&o\ar@{}[r]|\to&o\\
&&&\bq^-\\
&\bq^+	\ar@{.}[urr]
&\bq^+	\ar@{.}[ur]\\
\bq^-	\ar@{.}[ur]
}
\]
\vspace{-15pt}
\caption{Arena $(o \tto o) \tto o \tto o$}
\label{fig:ar_ex_arrow}
\end{minipage}
\end{figure}

We write $1$ for the empty arena and $o$ for the
arena with exactly one (negative) move. More elaborate types involve
matching constructions: the \emph{product} and
the \emph{arrow}.

\begin{definition}
Consider $A_1$ and $A_2$ arenas. Then, we define $A_1
\parallel A_2$ as
\[
\begin{array}{rcl}
\ev{A_1 \parallel A_2} &~~=~~& (\{1\} \times \ev{A_1}) \cup (\{2\}
\times \ev{A_2})\\
(i, a) \leq_{A_1 \parallel A_2} (j, b) &\Leftrightarrow&
	i = j ~~ \wedge ~~ a \leq_{A_i} b\\
\lambda_{A_1 \parallel A_2}(i, a) &=& \lambda_{A_i}(a)\,,
\end{array}
\]
called their \textbf{parallel composition} or
\textbf{product}, and also written $A_1 \times A_2$.
\end{definition}

For any family $(A_i)_{i\in I}$ of arenas, this extends to
$\prod_{i \in I} A_i$ in the obvious way. 
Any arena $A$ decomposes (up to forest
iso) as
$A \iso \prod_{i\in I} A_i$
for some family $(A_i)_{i\in I}$ of arenas which are
\textbf{well-opened}, \emph{i.e.} with \emph{exactly one}
initial (\emph{i.e. minimal})
move. We now define the \emph{arrow}: 

\begin{definition}
Consider $A_1, A_2$ arenas with $A_2$ well-opened. Then $A_1\tto A_2$ has:
\[
\begin{array}{rcl}
\ev{A_1\tto A_2} &~~=~~& (\{1\} \times \ev{A_1}) \cup (\{2\} \times
\ev{A_2})\\
(i, a) \leq_{A_1 \tto A_2} (j, b) &\Leftrightarrow& (i = j \wedge a
\leq_{A_i} b) \vee (i=2 \wedge a \in \min(A_2))\\
\lambda_{A_1 \tto A_2}(i, a) &=& (-1)^i \cdot \lambda_{A_i}(a)
\end{array}
\]

This extends to all arenas with $A \tto \prod_{i\in I} B_i =
\prod_{i\in I} A \tto B_i$ and $A \tto 1 = 1$.
\end{definition}

We will mostly use $A\tto B$ for $B$ well-opened. Figure
\ref{fig:ar_ex_arrow} displays $(o \tto o) \tto o \tto o$, matching the
simple type $(o \to o) \to o \to o$ with atomic type $o$ -- the position
of moves follows a correspondence between those and atoms
of the type. 
These arena constructions describe 
call-by-name computation: once Opponent initiates computation with
$\bq^-$, two Player moves become available. Player may call the second
argument (terminating computation) or evaluate the first argument, 
which in turn allows Opponent to call its argument.

\subsection{Plays and Strategies}

In Hyland-Ong games, players are allowed to \emph{backtrack}, and resume
the play from any earlier stage. This is made formal by the notion of
\emph{pointing strings}:

\begin{definition}
A \textbf{pointing string} over set $\Sigma$ is a string $s \in
\Sigma^*$, where each move may additionally come equipped with a
\textbf{pointer} to an earlier move.
\end{definition}

We often write $s = s_1 \dots s_n$ for pointing strings, leaving pointers
implicit. 

\begin{definition}
A \textbf{play} on arena $A$ is a pointing string $s = s_1 \dots s_n$ over
$\ev{A}$ s.t.:
\[
\begin{array}{rl}
\text{\emph{rigid:}} & \text{If $s_i$ points to $s_j$, then $s_j \imc_A
s_i$,}\\
\text{\emph{alternating:}} & \text{for all $1 \leq i < n$,
$\lambda_A(s_i) \neq \lambda_A(s_{i+1})$,}\\
\text{\emph{legal:}} & \text{for all $1 \leq i \leq n$, either $s_i \in
\min(A)$ or $s_i$ has a pointer.}
\end{array}
\]

A play is \textbf{well-opened} iff it has exactly one initial move. We
write $\Plays(A)$ for the set of plays on $A$, $\Plays^+(A)$ for
even-length plays, and $\Plays_\bullet(A)$ for 
well-opened plays.
\end{definition}

We write $\varepsilon$ for the empty play, $\prefix$ for the prefix, and
$\prefix^+$ if the smaller play has even length. 
\begin{figure}[t]
\begin{minipage}{.32\linewidth}
\[
\xymatrix@C=10pt@R=-1pt{
(o	\ar@{}[r]|\to&
o)	\ar@{}[r]|\to&
o	\ar@{}[r]|\to&
o\\
&&&\bq^-\\
&&\bq^+	\ar@{.}@/^/[ur]\\~\\~
}
\]
\caption{$\lambda f^{o\to o}.\,\lambda x^o.\,x$}
\label{fig:ex_play_1}
\end{minipage}
\hfill
\begin{minipage}{.32\linewidth}
\[
\xymatrix@C=10pt@R=-5pt{
(o	\ar@{}[r]|\to&
o)	\ar@{}[r]|\to&
o	\ar@{}[r]|\to&
o\\
&&&\bq^-\\
&\bq^+	\ar@{.}@/^/[urr]\\
\bq^-	\ar@{.}@/^/[ur]\\
&&\bq^+	\ar@{.}@/^/[uuur]
}
\]
\caption{$\lambda f^{o\to o}.\,\lambda x^o.\,f\,x$}
\label{fig:ex_play_2}
\end{minipage}
\hfill
\begin{minipage}{.32\linewidth}
\[
\xymatrix@C=10pt@R=-8.5pt{
(o	\ar@{}[r]|\to&
o)	\ar@{}[r]|\to&
o	\ar@{}[r]|\to&
o\\
&&&\bq^-\\
&\bq^+	\ar@{.}@/^/[urr]\\
\bq^-	\ar@{.}@/^/[ur]\\
&\bq^+	\ar@{.}@/^/[uuurr]\\
\bq^-	\ar@{.}@/^/[ur]\\
&&\bq^+	\ar@{.}@/^/[uuuuur]
}
\]
\caption{$\lambda f^{o\to o}.\,\lambda x^o.\,f\,(f\,x)$}
\label{fig:ex_play_3}
\end{minipage}
\end{figure}
Plays represent higher-order executions.
Figures \ref{fig:ex_play_1}, \ref{fig:ex_play_2} and \ref{fig:ex_play_3}
show plays on the arena of Figure \ref{fig:ar_ex_arrow}; 
matching typical executions of the corresponding simply-typed
$\lambda$-term. They are read from top to bottom, with 
pointers as dotted lines. As in Figure
\ref{fig:ar_ex_arrow}, the position of moves encodes their
identity in the arena.
Strategies, representing programs, are sets of plays:

\begin{definition}
A strategy $\sigma : A$ on arena $A$ is a non-empty set $\sigma
\subseteq \Plays^+(A)$ satisfying
\[
\begin{array}{rl}
\text{\emph{prefix-closed:}} & \forall s \in \sigma, \forall t \prefix^+
s, t \in \sigma\,,\\
\text{\emph{deterministic:}} & \forall s \in \sigma,~sab, sab' \in
\sigma \implies sab = sab'\,.
\end{array}
\]
\end{definition}

Implicit in the last clause is that $sab$ and $sab'$ also
have the same pointers.

\subsection{Visibility and Innocence}

\emph{Innocence} captures that the
behaviour only depends on which program phrase currently has
control. Intuitively, the ``current program phrase'' is captured by the
\emph{P-view}.

\begin{definition}
For any arena $A$, we set a partial function $\pview{-} : \Plays(A) \pto
\Plays(A)$ as:
\[
\begin{array}{rclcl}
\pview{si} &=& i &\qquad&\text{if $i \in \min(A)$,}\\
\pview{sn^-m^+} &=& \pview{sn}m &&\text{if the pointer of $m$ is in
$\pview{sn}$,}\\
\pview{sn^+tm^-} &=& \pview{sn}m && \text{if $m$ points to $n$,}
\end{array}
\]
undefined otherwise. In the last two cases, $m$ keeps its pointer in the resulting
play.

If defined, $\pview{s}$ is the \textbf{P-view} of $s$. A play $s
\in \Plays(A)$ is \textbf{visible} iff $\forall t \prefix s$,
$\pview{t}$ is defined.
\end{definition}

We say that $s \in \Plays(A)$ is a \textbf{P-view} iff
$\pview{s} = s$.  A strategy $\sigma : A$ is \textbf{visible} iff any $s
\in \sigma$ is visible. In that case, P-views are always
well-defined, so that we may formulate:

\begin{definition}\label{def:innocence}
A strategy $\sigma : A$ is \textbf{innocent} iff it is visible, and
satisfies:
\[
\begin{array}{rl}
\text{\emph{innocence:}} & \text{for all $sab, t \in \sigma$, if $ta \in
\Plays(A)$ and $\pview{sa} = \pview{ta}$, then $tab \in \sigma$.}
\end{array}
\]
where, in $tab$, $b$ points ``as in $sab$'', \emph{i.e.} so
as to ensure that $\pview{sab} = \pview{tab}$.
\end{definition}

An innocent $\sigma : A$ is determined by $\pviews{\sigma} =
\{\pview{s} \mid s \in \sigma\}$, its 
\emph{P-view forest}. Figures
\ref{fig:ex_play_1}, \ref{fig:ex_play_2} and \ref{fig:ex_play_3}
present P-views, each inducing an innocent strategy \emph{via} the
P-view forest obtained by even-length prefix
closure. 
\begin{figure}[t]
\begin{minipage}{.45\linewidth}
\[
\scalebox{.9}{$
\xymatrix@C=2pt@R=-8pt{
((o & \to & o) &\to& o) &\to &o\\
&&&&&& \bq^-  
	\ar@{.}@/_/[dll]
	\ar@{.}@/_/[dddll] \\
&&&&\bq^+ 
	\ar@{.}@/_/[dll]&& \\
&&\bq^- \ar@{.}@/_/[dddll]  & & & &\\     
&&&&\bq^+ 
	\ar@{.}@/_/[dll] & & \\
&&\bq^-\\
\bq^+
}$}
\]
\caption{$K_x : ((o \to o) \to o)\to o$}
\label{fig:ex_str_kx}
\end{minipage}
\hfill
\begin{minipage}{.45\linewidth}
\[
\scalebox{.9}{$
\xymatrix@C=2pt@R=-8pt{
((o & \to & o) &\to& o) &\to &o\\
&&&&&& \bq^-  
	\ar@{.}@/_/[dll]
	\ar@{.}@/_/[dddll] \\
&&&&\bq^+ 
	\ar@{.}@/_/[dll]&& \\
&&\bq^- \\
&&&&\bq^+ 
	\ar@{.}@/_/[dll] & & \\
&&\bq^- \ar@{.}@/_/[dll]  \\
\bq^+
}$}
\]
\caption{$K_y : ((o \to o) \to o) \to o$}
\label{fig:ex_str_ky}
\end{minipage}
\end{figure}
Likewise, Figures \ref{fig:ex_str_kx} and \ref{fig:ex_str_ky} induce
strategies for the so-called simply-typed ``Kierstead terms'' $\lambda
f^{(o \to o)\to o}.\,f\,(\lambda x^o.\,f\,(\lambda y^o.\,x))$ and
$\lambda f^{(o \to o) \to o}.\,f\,(\lambda x^o.\,f\,(\lambda y^o.\,y))$.
P-views are well-opened, so innocent strategies are determined
by their set $\sigma_\bullet$ of well-opened plays.

Innocent strategies form a cartesian closed category
$\Inn$ with as objects arenas, and morphisms from $A$ to $B$ 
the innocent strategies $\sigma : A \tto B$. Composing $\sigma : A \tto B$
and $\tau : B \tto C$ involves a ``parallel interaction plus hiding''
mechanism, which we omit \cite{HylandO00}.

\subsection{Positions}

Boudes' ``thick subtrees'' \cite{DBLP:conf/tlca/Boudes09}, called
\emph{positions} in this paper, are the central concept informing the
link between innocent game semantics and relational semantics. They are
simply desequentialized plays, or in other words prefixes of the
arena with duplications.

To introduce positions, our first stop is the following notion of
\emph{configuration}.

\begin{definition}\label{def:conf}
A \textbf{configuration} $x \in \conf{A}$ of arena $A$ is a tuple $x =
\tuple{\ev{x}, \leq_x, \partial_x}$ such that $\tuple{\ev{x}, \leq_x}$ is a
finite tree, and $\partial_x : \ev{x} \to \ev{A}$,
the \textbf{display map}, is a labeling function s.t.:
\[
\begin{array}{rl}
\text{\emph{minimality-respecting:}} & \text{for all $a \in \ev{x}$, $a$
is $\leq_x$-minimal iff $\partial_x(a)$ is $\leq_A$-minimal,}\\
\text{\emph{causality-preserving:}} & \text{for all $a_1, a_2 \in
\ev{x}$, if $a_1 \imc_x a_2$ then $\partial_x(a_1) \imc_A \partial_x(a_2)$.}
\end{array}
\]
\end{definition}

We call \textbf{events} the elements of $\ev{x}$. Note
$\tuple{\ev{x}, \leq_x}$ has exactly one minimal event, which suffices as
innocent strategies are determined by well-opened plays.  Configurations
include:

\begin{definition}\label{def:deseq_plays}
The \textbf{desequentialization} $\deseq{s} \in \conf{A}$ of
$s = s_1 \dots s_n \in \Plays_\bullet(A)$ has 
$\ev{\deseq{s}} = \{1, \dots, n\}$, $\partial_{\deseq{s}}(i) = s_i$, and $i
\leq_{\deseq{s}} j$ if there is a chain of
pointers from $s_j$ to $s_i$ in $s$.
\end{definition}

\begin{figure}[t]
\begin{minipage}{.4\linewidth}
\[
\scalebox{.9}{$
        \begin{array}{ccc}
        \xymatrix@C=-1pt@R=1pt{
        \bq^- \ar@{.}[d] &&&&& 
                & 1 \ar@{.}[dr] \ar@{.}[dl] \ar@{|->}[llllll]_{\partial_x}& \\
        \bq^+ \ar@{.}[d] &&&&& 
                2 \ar@{.}[d] \ar@{|->}[lllll] 
                && 4 \ar@{.}[d]
\ar@{|->}@/_0.7pc/[lllllll] \\
        \bq^- \ar@{.}[d] &&&&&
                3 \ar@{.}[d]\ar@{|->}[lllll] && 5
\ar@{|->}@/_0.7pc/[lllllll] \\
        \bq^+ &&&&& 
                6 \ar@{|->}[lllll]
        }
        & \quad & 
        \xymatrix@C=-1pt@R=1pt{
                        & 1 \ar@{.}[dr] \ar@{.}[dl]
\ar@{|->}[rrrrrr]^{\partial_y} & 
                &&&&& \bq^- \ar@{.}[d] \\
                        2 \ar@{.}[d] \ar@{|->}@/^0.7pc/[rrrrrrr] && 4
\ar@{.}[d] \ar@{|->}[rrrrr]
                &&&&& \bq^+ \ar@{.}[d] \\
                        3  \ar@{|->}@/^0.7pc/[rrrrrrr] && 5 \ar@{.}[d]
\ar@{|->}[rrrrr]
                &&&&& \bq^- \ar@{.}[d]\\
                        && 6 \ar@{|->}[rrrrr]
                &&&&& \bq^+ 
        }
        \end{array}
$}
\]
\caption{Deseq. $K_x$ and $K_y$}
\label{fig:ex_deseq_k}
\end{minipage}
\hfill
\begin{minipage}{.5\linewidth}
\[
\scalebox{.9}{$
\xymatrix@C=0pt@R=-6.5pt{
(o	\ar@{}[r]|\to&
o	\ar@{}[r]|\to&
o)	\ar@{}[r]|\to&
o	\ar@{}[r]|\to&
o\\
&&&&\bq^-\\
&&\bq^+	\ar@{.}@/^/[urr]\\
\bq^-	\ar@{.}@/^/[urr]\\
&&\bq^+	\ar@{.}@/^/[uuurr]\\
&\bq^-	\ar@{.}@/^/[ur]\\
&&&\bq^+\ar@{.}@/^/[uuuuur]
}
\xymatrix@C=0pt@R=-6.5pt{
(o      \ar@{}[r]|\to&
o       \ar@{}[r]|\to&
o)      \ar@{}[r]|\to&
o       \ar@{}[r]|\to&
o\\
&&&&\bq^-\\
&&\bq^+ \ar@{.}@/^/[urr]\\
&\bq^-   \ar@{.}@/^/[ur]\\
&&\bq^+ \ar@{.}@/^/[uuurr]\\
\gray{\bq^-}
	\ar@{.}@/^/@[gray][urr]\\
~
}$}
\]
\caption{Non-positionality of innocence}
\label{fig:non_pos_inn}
\end{minipage}
\end{figure}
We show in Figure \ref{fig:ex_deseq_k} the desequentialization of the
maximal P-views of $K_x$ and $K_y$ from Figures
\ref{fig:ex_str_kx} and \ref{fig:ex_str_ky}. 
Extracting $\deseq{s}$ is a first step, we must then
forget the identity of its events:

\begin{definition}
A bijection $\varphi : \ev{x} \iso \ev{y}$ is an \textbf{isomorphism}
$\varphi : x \iso y$ iff it is
\[
\begin{array}{rl}
\text{\emph{arena-preserving:}} & \text{for all $a \in \ev{x}$,
$\partial_y(\varphi(a)) = \partial_x(a)$,}\\
\text{\emph{causality-respecting:}} & \text{for all $a_1, a_2 \in
\ev{x}$, we have $a_1 \imc_x a_2$ iff $\varphi(a_1) \imc_y
\varphi(a_2)$.}
\end{array}
\]

A \textbf{position} of $A$, written $\x \in \pos{A}$, is an
isomorphism class of configurations.
\end{definition}

If $s \in \Plays_\bullet(A)$, the \textbf{position} $\coll{s} \in
\pos{A}$ is the isomorphism class of $\deseq{s}$.

We pause to consider the \emph{positionality of innocent
strategies} as mentioned in the introduction. Though it will only play a
very minor role, we define \emph{positional} strategies:

\begin{definition}
Consider $\sigma : A$ a strategy on $A$. We set the condition:
\[
\begin{array}{rl}
\text{\emph{positional:}} &
\text{$\forall sab, t \in \sigma$, $ta \in \Plays(A)$, $\coll{sa}
= \coll{ta} \implies \exists tab \in \sigma, \coll{sab} = \coll{tab}$.}
\end{array}
\]
\end{definition}

Innocent strategies are not
positional: Figure \ref{fig:non_pos_inn} displays (the two maximal
P-views of) the innocent strategy for the $\lambda$-term $\lambda
f^{o \to o \to o}.\,\lambda x^o.\,f\,(f\,\bot\,x)\,(f\,\bot\,\bot)$. On
the right hand side, the last Opponent move is grayed out as an
extension of a P-view triggering no response.
After the fifth move the position is the same, 
contradicting positionality. 
In Melli\`es' asynchronous games
\cite{DBLP:journals/tcs/Mellies06}, explicit copy
indices help distinguish the two calls to $f$. The two plays no longer
reach the same position, restoring positionality.
But even in asynchronous games, if positions were quotiented by symmetry
so as to match relational semantics, positionality would fail.

We turn to the weaker \emph{positional injectivity}.  If $\sigma : A$,
its \textbf{positions} are those reached by well-opened plays,
\emph{i.e.} $ \positions{\sigma} = \{ \coll{s} \mid s \in
\sigma_\bullet\} \subseteq \pos{A} $. We may finally ask our main
question:

\begin{quest}[Positional Injectivity]\label{q:1}
If $\sigma, \tau$ are innocent and $\positions{\sigma} = \positions{\tau}$,
do we have $\sigma = \tau$? 
\end{quest}

\subsection{Links with the Relational Model}
\label{subsec:link_rel}

To fully appreciate this question, it is informative to consider the
link with the relational model. We start with the following observation
concerning positions on the arrow arena.

\begin{fact}\label{fact:rel}
Consider $A$ and $B$ arenas, and write $\mathcal{M}_f(X)$ for the
\textbf{finite multisets} on $X$.

Then, we have a bijection $\pos{A \tto B} \iso \mathcal{M}_f(\pos{A})
\times \pos{B}$.
\end{fact}

Recall \cite{DBLP:journals/tcs/Ehrhard12} that the relational model
forms a cartesian closed category $\Rel_\oc$ having \emph{sets} as
objects; and as morphisms from $A$ to $B$ the \emph{relations} $R
\subseteq \mathcal{M}_f(A) \times B$.
Considering simple types
generated from $o$ and the arrow $A \to B$, and
setting the relational interpretation of $o$ as $\intr{o}_{\Rel_\oc} =
\{\bq\}$, then for any type $A$, there is a
bijection $r_A : \pos{\intr{A}_{\Inn}} \iso \intr{A}_{\Rel_\oc}$.  

\begin{theorem}
This extends to a functor $\positions{-} : \Inn \to
\Rel_\oc$, which preserves the interpretation: for any term $M
: A$ of the simply-typed $\lambda$-calculus, 
$r_A(\positions{\intr{M}_\Inn}) = \intr{M}_{\Rel_\oc}$.
\end{theorem}

This \emph{relational collapse} of innocent strategies
has been studied extensively
\cite{DBLP:conf/tlca/Boudes09,DBLP:journals/tcs/Mellies06,DBLP:conf/lics/TsukadaO16,DBLP:conf/lics/CastellanCPW18,DBLP:journals/pacmpl/ClairambaultV20}. The inclusion $\subseteq$ is easy; the
difficulty in proving $\supseteq$ is that game-semantic interaction is
temporal: positions arising relationally
might, in principle, fail to appear game-semantically 
because reproducing them yields a deadlock.
For innocent strategies this does not happen:
this may be proved through connections with syntax
\cite{DBLP:conf/tlca/Boudes09,DBLP:conf/lics/TsukadaO16} or semantically
\cite{DBLP:conf/lics/CastellanCPW18,DBLP:journals/pacmpl/ClairambaultV20}.

In \cite{DBLP:conf/lics/TsukadaO16}, Tsukada and Ong
prove a similar collapse injective.
This seems to answer Question \ref{q:1} positively -- but this is not so
simple. The interpretation in $\Rel_\oc$ is parametrized by a
set $X$ for the ground type $o$. In
\cite{DBLP:conf/lics/TsukadaO16}, $X$ is required to be countably
infinite: this way one allocates one tag for each pair of chronologically
contiguous O/P
moves, encoding the \emph{causal} / \emph{axiom links}.  In contrast,
for Question \ref{q:1} we are forced to interpret
$o$ with a singleton set $\{\bq\}$, or lose the correspondence between
points of the web and positions.  We must reconstruct strategies
directly from their desequentializations, with no help from labeling or
coloring. 

\subsection{Main result}
\label{subsec:main_result}

At first this seems desperate. In
\cite{DBLP:conf/lics/TsukadaO16}, an innocent strategy 
may already be reconstructed from the desequentialization of its
P-views. But here, the two plays of
Figures \ref{fig:ex_str_kx} and \ref{fig:ex_str_ky} yield the 
configurations of Figure \ref{fig:ex_deseq_k}, which are
isomorphic -- so give the same position. Nevertheless $K_x$ and
$K_y$ \emph{can} be distinguished, via their behaviour under
replication.
\begin{figure}[t]
\[
\scalebox{.9}{$
\xymatrix@C=1pt@R=-6pt{
&& ((o & \to & o) & \to & o) & \to & o
&&&&&&
((o & \to & o) & \to & o) & \to & o \\
&&&&&&&& \bq^- \ar@{.}@/_/[dll]  \ar@{.}@/_/[dddll]
&&&&&&&&&&&& \bq^-  \ar@{.}@/_/[dll] \ar@{.}@/_/[dddll] \\
&&&&&& \bq^+ \ar@{.}@/_/[dll]
&&&&&&&&&&&& \bq^+ \ar@{.}@/_/[dll] \\
K_x\ni &&
&& {\bq^-} \ar@{.}@/_/[dddll]
&&&&&&&& K_y\ni &&&& \bq^- \\
&&&&&& \bq^+ \ar@{.}@/_/[dll]
&&&&&&&&&&&& \bq^+ \ar@{.}@/_/[dll] \\
&&&& \bq^-
&&&&&&&&&&&& {\bq^-} \ar@{.}@/_/[dll] \\
&& \bq^+
&&&&&&&&&&&& \bq^+ \\
&&&& {\bq^-} \ar@{.}@/^/[uuuuurr]
&&&&&&&&&&&& {\bq^-} \ar@{.}@/^/[uuurr] \\
&&&&&& {\bq^+} \ar@{.}@/^/[uuuuuuurr]
&&&&&&&& {\bq^+} \ar@{.}@/^/[urr]
}$}
\]
\caption{Plays yielding positions distinguishing $K_x$ and $K_y$}
\label{fig:play_kx_ky}
\end{figure}
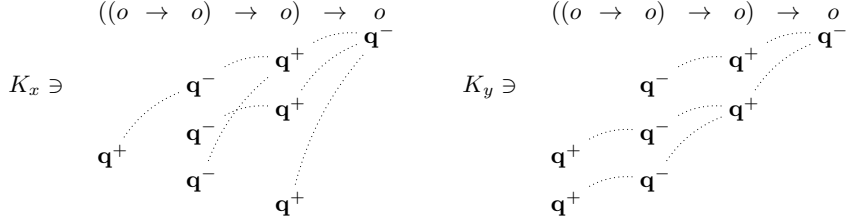
In both plays of Figure \ref{fig:play_kx_ky}, we replay the move to
which the deepest $\bq^+$ points. This brings $K_x$ and $K_y$ to react
differently, obtaining plays whose positions separate
$\positions{K_x}$ and $\positions{K_y}$. So, by observing
the behaviour of a strategy under replication, we can infer some
temporal information.

Most of the paper will be devoted to turning this idea into a proof.
However, we have only been able to prove the result with the following
additional restrictions on strategies.

\begin{definition}
For $A$ an arena, we define conditions on innocent strategies
$\sigma : A$ as:
\[
\begin{array}{rl}
\text{\emph{total:}} & \text{for all $s \in \sigma$, if $sa \in
\Plays(A)$ then there exists $b$ such that $sab \in \sigma$,}\\
\text{\emph{finite:}} & \text{the set $\pviews{\sigma} = \{\pview{s}
\mid s \in \sigma\}$ is finite.}
\end{array}
\]
\end{definition}

Total finite strategies are already well-known: on arenas interpreting
simple types they exactly correspond to $\beta$-normal $\eta$-long
normal forms of simply-typed $\lambda$-terms.  

We now state our main result, \textbf{positional injectivity}:

\begin{theorem}\label{th:main}
For any $\sigma, \tau : A$ innocent total finite, 
$\sigma = \tau$ iff $\positions{\sigma} = \positions{\tau}$.
\end{theorem}

As observed in Section \ref{sect:gs-arena}, all arenas decompose as $A =
\prod_{i\in I} A_i$ with $A_i$ well-opened. As $\times$
is a cartesian product in $\Inn$, strategies $\sigma : A$ 
also decompose as $\sigma = \tuple{\sigma_i \mid i \in I}$ with
$\sigma_i : A_i$ for all $i\in I$. From innocence it follows
that 
$
\positions{\tuple{\sigma_i \mid i \in I}} \iso \sum_{i\in I}
\positions{\sigma_i}
$,
so it suffices to prove Theorem \ref{th:main} for $A$ well-opened.
From now on, we consider all arenas well-opened.

\section{Causal Presentation}\label{sec:caus}

Besides the behaviour of strategies under replication, plays also
include the order, irrelevant for our purposes, in which branches are
explored by Opponent.
To isolate the effect of replication, we introduce a \emph{causal}
version of strategies inspired from \emph{concurrent games}
\cite{cg1}.

\subsection{Augmentations}

This formulation rests on the notion of \emph{augmentations}.
Intuitively those correspond to expanded trees of P-views, which
enrich configurations with causal wiring from the strategy.

\begin{definition}\label{def:augmentation}
An \textbf{augmentation} on arena $A$ is a tuple $\q = \tuple{\ev{\q},
\leq_{\deseq{\q}}, \leq_\q, \partial_\q}$, where $\deseq{\q} =
\tuple{\ev{\q}, \leq_{\deseq{\q}},\partial_\q} \in \conf{A}$, and
$\tuple{\ev{\q}, \leq_\q}$ is a tree satisfying:
\[
\begin{array}{rl}
\text{\emph{rule-abiding:}} & \text{for all $a_1, a_2 \in \ev{\q}$, if
$a_1 \leq_{\deseq \q} a_2$, then $a_1 \leq_\q a_2$,}\\ 
\text{\emph{courteous:}} & \text{for all $a_1 \imc_\q a_2$, if
$\lambda(a_1) = +$ or $\lambda(a_2) = -$, then $a_1 \imc_{\deseq \q} a_2$,}\\
\text{\emph{deterministic:}} & \text{for all $a^- \imc_\q a^+_1$ and
$a^- \imc_\q a^+_2$, then $a_1 = a_2$,}
\end{array}
\]
we then write $\q \in \Aug(A)$, and call
$\deseq{\q} \in \conf{A}$ the \textbf{desequentialization} of $\q$.
\end{definition}

Events of $\ev{\q}$ inherit a polarity with
$\lambda(a) = \lambda_A(\partial_\q(a))$. By \emph{rule-abiding} and
\emph{courteous}, $\tuple{\ev{\q}, \leq_\q}$ and $\tuple{\ev{\q},
\leq_{\deseq{\q}}}$ have the same minimal event $\init(\q)$, called the
\textbf{initial} event. If $a \in \ev{\q}$ is not initial, there is
a unique $a' \in \ev{\q}$ such that $a' \imc_\q a$, written $a' =
\pred(a)$ and called the \textbf{predecessor} of $a$. Likewise, a
non-initial $a \in \ev{\q}$ also has a unique $a'' \in \ev{\q}$ such
that $a'' \imc_{\deseq{\q}} a$, written $a'' = \just(a)$ and called the 
\textbf{justifier} of $a$. 
By \emph{courteous} and as immediate
causality alternates in $A$ (and hence in $\deseq{\q}$), both
$\pred(a)$ and $\just(a)$ have polarity opposite to $a$. They may
not coincide, however from \emph{courteous} they do for
$a$ negative.

\begin{figure}[t]
\begin{minipage}{.46\linewidth}
\[
\left(
\raisebox{27pt}{$
\scalebox{.85}{$
\xymatrix@C=-2pt@R=5pt{
&&&\bq^-\ar@{.}[dll]
	\ar@{.}[d]
	\ar@{.}[drr]
	\ar@{-|>}@/_.1pc/[dll]\\
&\bq^+	\ar@{.}[dl]
	\ar@{.}[dr]
	\ar@{-|>}@/_.1pc/[dl]
	\ar@{-|>}@/_.1pc/[dr]&&
\bq^+	\ar@{.}[d]
	\ar@{-|>}@/_.1pc/[d]&&
\bq^+	\ar@{.}[d]
	\ar@{-|>}@/_.1pc/[d]\\
\bq^-	\ar@{.}[d]
	\ar@{-|>}[urrr]&&
\bq^-	\ar@{.}[d]
	\ar@{-|>}[urrr]&
\bq^-	\ar@{-|>}[dlll]&&
\bq^-	\ar@{-|>}[dlll]\\
\bq^+&&\bq^+
}$}$}
\right)
\in \exp
\left(
\raisebox{27pt}{$
\scalebox{.85}{$
\xymatrix@C=-4pt@R=5pt{
&\bq^-	\ar@{.}[dl]
	\ar@{.}[dr]
	\ar@{-|>}@/_.1pc/[dl]\\
\bq^+	\ar@{.}[d]
	\ar@{-|>}@/_.1pc/[d]&&
\bq^+	\ar@{.}[d]
	\ar@{-|>}@/_.1pc/[d]\\
\bq^-	\ar@{.}[d]
	\ar@{-|>}[urr]&&
\bq^-	\ar@{-|>}[dll]\\
\bq^+
}$}$}
\right)
\]
\caption{Causal $K_x$ and its expansion}
\label{fig:caus_kx_exp}
\end{minipage}
\hfill
\begin{minipage}{.46\linewidth}
\[
\left(
\raisebox{27pt}{$
\scalebox{.9}{$
\xymatrix@C=-4pt@R=4.5pt{
&&\bq^-	\ar@{.}[dll]
	\ar@{.}[drr]
	\ar@{-|>}@/_.1pc/[dll]\\
\bq^+	\ar@{.}[d]
	\ar@{-|>}@/_.1pc/[d]&&&&
\bq^+	\ar@{.}[dl]
	\ar@{-|>}@/_.1pc/[dl]
	\ar@{.}[dr]
	\ar@{-|>}@/_.1pc/[dr]\\
\bq^-	\ar@{-|>}[urrrr]&&&
\bq^-	\ar@{.}[d]
	\ar@{-|>}@/_.1pc/[d]&&
\bq^-	\ar@{.}[d]
	\ar@{-|>}@/_.1pc/[d]\\
&&&\bq^+&&\bq^+
}
$}$}
\right)
\in \exp
\left(
\raisebox{27pt}{$
\scalebox{.9}{$
\xymatrix@C=-4pt@R=4.5pt{
&\bq^-	\ar@{.}[dl]
	\ar@{.}[dr]
	\ar@{-|>}@/_.1pc/[dl]\\
\bq^+	\ar@{.}[d]
	\ar@{-|>}@/_.1pc/[d]&&
\bq^+	\ar@{.}[d]
	\ar@{-|>}@/_.1pc/[d]\\
\bq^-	\ar@{-|>}[urr]&&
\bq^-	\ar@{.}[d]
	\ar@{-|>}@/_.1pc/[d]\\
&&\bq^+
}
$}$}
\right)
\]
\caption{Causal $K_y$ and its expansion}
\label{fig:caus_ky_exp}
\end{minipage}
\end{figure}
Figures \ref{fig:caus_kx_exp} and \ref{fig:caus_ky_exp} show
augmentations -- though the corresponding definitions remain to be seen,
those are the causal expansions of $K_x$ and $K_y$ matching the
plays of Section \ref{subsec:main_result}. In such diagrams, 
immediate causality from the configuration appears as
dotted lines, whereas that coming from the augmentation itself appears
as $\imc$.
We set a few auxiliary conditions:

\begin{definition}
Let $\q \in \Aug(A)$ be an augmentation. We set the conditions:
\[
\begin{array}{rl}
\text{\emph{receptive:}}
& \text{for all $a \in \ev{\q}$, if $\partial_{\q}(a)\imc_A b^-$,
there is $a \imc_\q b'$ such that $\partial_{\q}(b') = b$,}\\
\text{\emph{$+$-covered:}}
& \text{for all $a \in \ev{\q}$ maximal in $\q$, we have $\lambda(a)=+$,}\\
\text{\emph{$-$-linear:}} & \text{for all $a \imc_\q a_1^-$, $a \imc_\q
a_2^-$, if $\partial_\q(a_1) = \partial_\q(a_2)$ then $a_1 = a_2$.}
\end{array}
\]
\end{definition}

We say that $\q \in \Aug(A)$ is \textbf{total} iff it is receptive
and $+$-covered. We will also refer to receptive $-$-linear
augmentations as \textbf{causal strategies}.

\subsection{From Strategies to Causal Strategies}\label{subsec:strat-to-caus}

We may easily represent an innocent strategy as a causal strategy:

\begin{proposition}\label{prop:pview_aug}
For $\sigma : A$ finite innocent on $A$ well-opened, we set components
\[
\ev{\caus{\sigma}} = \{\pview{s} \mid s \in \sigma \wedge s \neq
\varepsilon\} \cup \{\pview{sa} \mid s \in \sigma \wedge sa \in
\Plays(A)\}\,,
\]
$s \leq_{\caus{\sigma}} t$ iff $s \prefix t$, $sa
\leq_{\deseq{{\caus{\sigma}}}} satb$ iff there is a chain of justifiers
from $b$ to $a$, and $\partial_{\caus{\sigma}}(sa) = a$.

Then ${\caus{\sigma}} = \tuple{\ev{{\caus{\sigma}}}, \leq_{\caus{\sigma}},
\leq_{\deseq{{\caus{\sigma}}}}, \partial_{\caus{\sigma}}} \in \Aug(A)$ is a
causal strategy, and is total iff $\sigma$ is total.
\end{proposition}

The proof appears as Lemmas \ref{lem:pview_aug} and \ref{lem:sig_tot_caus_tot} in Appendix \ref{app:caus}.
As for configurations, so as to forget the concrete identity of events
we consider augmentations up to \emph{isomorphism}:

\begin{definition}\label{def:morphism}
A \textbf{morphism} $\varphi : \q \to \p$ is a function
$\varphi : \ev{\q} \to \ev{\p}$ satisfying:
\[
\begin{array}{rl}
\text{\emph{arena-preserving:}} & \text{$\partial_\p \circ \varphi = \partial_\q$,}\\
\text{\emph{causality-preserving:}} & \text{for all $a_1, a_2 \in
\ev{\q}$, if $a_1 \imc_\q a_2$ then $\varphi(a_1) \imc_\p
\varphi(a_2)$,}\\
\text{\emph{configuration-preserving:}} & \text{for all $a_1, a_2 \in
\ev{\q}$, if $a_1 \imc_{\deseq \q} a_2$ then $\varphi(a_1) \imc_{\deseq
\p} \varphi(a_2)$.} 
\end{array} 
\]

An \textbf{isomorphism} is an invertible morphism -- we then write
$\varphi : \q \iso \p$.
\end{definition}

Note that by \emph{arena-preserving}, $\varphi$ must send $\init(\q)$ to
$\init(\p)$.

The reader may check that the construction of Proposition
\ref{prop:pview_aug} applied to $K_x$ and $K_y$ yields, up to
isomorphism, the (small) augmentations of Figures \ref{fig:caus_kx_exp}
and \ref{fig:caus_ky_exp}. 
The next fact shows that augmentations are indeed an alternative
presentation of innocent strategies.

\begin{lemma}\label{lem:strat_eq_caus_iso}
For any finite innocent strategies $\sigma,\tau$ on arena $A$, then 
$\sigma = \tau$ iff ${\caus{\sigma}} \iso {\caus{\tau}}$.
\end{lemma}
\begin{proof}
Clearly, $\sigma = \tau$ implies ${\caus{\sigma}} = {\caus{\tau}}$.
Conversely, assume $\varphi:{\caus{\sigma}} \iso {\caus{\tau}}$. Take $s
= s_1 \dots s_n \in\pviews{\sigma}$, and write $s_{\leq i} = s_1 \dots
s_i$. Then we have a chain $s_{\leq 1} \imc_{{\caus{\sigma}}} s_{\leq 2}
\imc_{{\caus{\sigma}}} \ldots  \imc_{{\caus{\sigma}}} s_{\leq n-1}
\imc_{{\caus{\sigma}}} s$, transported through $\varphi$ to
$t_{\leq 1} \imc_{{\caus{\tau}}} \ldots  \imc_{{\caus{\tau}}} t$. By
\emph{arena-preserving}, $t_i = s_i$ for all $1\leq i \leq n$.  Finally
by \emph{configuration-preserving}, $s$ and $t$ have the same
pointers, hence $s=t$ and $s \in \tau$.  Symmetrically, any
P-view $t \in \pviews{\tau}$ is in $\sigma$, hence $\pviews{\sigma} =
\pviews{\tau}$ and $\sigma = \tau$ by innocence.
\end{proof}

\subsection{Expansions of Causal Strategies}

Besides including representations of innocent strategies, augmentations
can also represent their \emph{expansions}, \emph{i.e.} arbitrary plays,
with Opponent's scheduling factored out. 

\begin{definition}\label{def:exp}
Consider $A$ an arena, and $\p \in \Aug(A)$ a causal strategy.

An \textbf{expansion} of $\p$, written $\q \in \exp(\p)$, is $\q \in
\Aug(A)$ such that:
\[
\begin{array}{rl}
\text{\emph{simulation:}} & \text{there is a (necessarily unique)
morphism $\varphi : \q \to \p$,}\\
\text{\emph{$+$-obsessional:}} & \text{for all $a^- \in \ev{\q}$ and
$\varphi(a^-) \imc_\p b^+$, there is $a^- \imc_\q a'$ s.t.
$\varphi(a') = b^+$.}
\end{array}
\]
\end{definition}

The relationship between a causal strategy $\p$ and $\q \in
\exp(\p)$ is analogous to that between an arena $A$ and a configuration
$x \in \conf{A}$: $\q$ explores a prefix of $\p$, possibly visiting
the same branch many times. However, \emph{determinism} ensures that
only Opponent may cause duplications, and \emph{$+$-obsessional} ensures
that only Opponent may refuse to explore certain branches -- if a Player
move is available in $\p$, then it must appear in all
corresponding branches of $\q$.
Uniqueness of the morphism follows from $-$-linearity and determinism
-- see Lemma \ref{lem:mor_unique} in Appendix \ref{app:caus}. Figures
\ref{fig:caus_kx_exp} and \ref{fig:caus_ky_exp} show expansions of (the
causal strategies corresponding to) $K_x$ and $K_y$.  

Now, we set
$
\positions{\p} = \{\coll{\q} \mid \q \in \exp(\p)\}
$
the \textbf{positions} of a causal strategy $\p$, 
where $\coll{\q}$ is the isomorphism class of $\deseq{\q}$.
By Lemma \ref{lem:strat_eq_caus_iso},
any innocent $\sigma : A$ yields a causal strategy ${\caus{\sigma}} :
A$, so this leaves us with the task to prove that the two notions of
position coincide.

\begin{restatable}{proposition}{possigposcaus}\label{prop:pos_sig_pos_caus}
For any total finite innocent strategy $\sigma : A$, we have $\positions{\sigma} =
\positions{{\caus{\sigma}}}$.
\end{restatable}
\begin{proof}
Additional details appear in Appendix \ref{app:caus}.
Any $\x \in \positions{\sigma}$ is the isomorphism class of $\deseq{s}$ for
$s = s_1 \dots s_n \in \sigma$. We build an expansion $\q(s) \in
\exp({\caus{\sigma}})$ as follows. Its configuration is
$\deseq{\q(s)} = \deseq{s}$ (see Definition \ref{def:deseq_plays}) with
events $\ev{\q(s)} = \{1, \dots, n\}$.  Its causal order is $i
\leq_{\q(s)} j$ iff $j \geq i$ and $s_i$ is reached in the computation
of $\pview{s_{\leq j}}$. 
To show that $\q(s) \in \exp({\caus{\sigma}})$ we must provide
a morphism $\varphi : \q(s) \to {\caus{\sigma}}$, which is simply $\varphi(i) =
\pview{s_{\leq i}}$. So, $\x = \coll{\q(s)} \in \positions{{\caus{\sigma}}}$.

Reciprocally, take $\x \in \positions{{\caus{\sigma}}}$, obtained as the isomorphism
class of some $\deseq{\q}$, for $\q \in \exp({\caus{\sigma}})$. From the
totality of $\sigma$, $\q$ has maximal events all
positive -- it has exactly as many Player as Opponent events, and 
admits a linear extension $s = s_1 \dots s_n$ which is
\emph{alternating}, \emph{i.e.} $\lambda(s_i) \neq
\lambda(s_{i+1})$ for all $1 \leq i \leq n-1$. Besides, for any $1\leq i
\leq n$, $\pview{s_{\leq i}}$ (treating $s$ as a play on
arena $\deseq{\q}$) coincides with $[s_i]_\q = \{s \in \ev{\q} \mid s
\leq_\q s_i\}$, totally ordered by $\leq_\q$. So, writing
$\partial_\q(s) = \partial_\q(s_1) \dots \partial_\q(s_n) \in \Plays(A)$
with pointers inherited from $\deseq{\q}$, 
$\pview{\partial_\q(s)_{\leq i}} \in \pviews{\sigma}$, hence
$\partial_\q(s) \in \sigma$ by innocence and $\deseq{\partial_\q(s)}
\iso \deseq{s}$. Therefore, $\coll{\q} = \coll{\partial_\q(s)} \in
\positions{\sigma}$. 
\end{proof}

The idea is that plays in $\sigma$ are
exactly linearizations of expansions of ${\caus{\sigma}}$. From a play
we get an expansion by factoring out Opponent's scheduling, mimicking
the construction of P-views while keeping duplicated branches
separate. Reciprocally, an expansion allows many (alternating)
linearizations. For instance, 
the two plays of Section \ref{subsec:main_result} are respectively
linearizations of the expansions of Figures \ref{fig:caus_kx_exp}
and \ref{fig:caus_ky_exp}.
This proposition fails if $\sigma$ is not total, as expansions may then
have trailing Opponent moves, preventing an alternating linearization.

Thanks to Proposition \ref{prop:pos_sig_pos_caus}, we focus on
positions reached by expansions of causal strategies. 

\section{Positional Injectivity}\label{sec:pos_inj}

We now come to the main contribution of this paper, the proof of
positional injectivity for total finite causal strategies. We start this
section by introducing the proof idea.

\subsection{Forks and Characteristic Expansions}
\label{subsec:char_exp}

Just from the static snapshot offered by positions, we must deduce the
strategy. 

Given $z \in \conf{A}$, can we uniquely reconstruct its \emph{causal
explanation}, \emph{i.e.} $\q \in \Aug(A)$ such that $z = \deseq{\q}$? In
general, there is no reason why $\q$ would be uniquely determined. 
\begin{figure}[t]
\begin{minipage}{.47\linewidth}
\[
\raisebox{20pt}{$
\scalebox{.75}{$
\xymatrix@C=-5pt@R=5pt{
&&&\bq^-\ar@{.}[dll]
        \ar@{.}[d]
        \ar@{.}[drr]
\\
&\bq^+  \ar@{.}[dl]
        \ar@{.}[dr]
&&
\bq^+   \ar@{.}[d]
&&
\bq^+   \ar@{.}[d]
\\
\bq^-   \ar@{.}[d]
&&
\bq^-   \ar@{.}[d]
&
\bq^-   
&&
\bq^-   
\\
\bq^+&&\bq^+
}$}$}
\!\!
\leadsto
\!\!
\raisebox{20pt}{$
\scalebox{.75}{$
\xymatrix@C=-4pt@R=5pt{
&&&\bq^-\ar@{.}[dll]
        \ar@{.}[d]
        \ar@{.}[drr]
        \ar@{-|>}@/_.1pc/[dll]\\
&\bq^+  \ar@{.}[dl]
        \ar@{.}[dr]
        \ar@{-|>}@/_.1pc/[dl]
        \ar@{-|>}@/_.1pc/[dr]&&
\bq^+   \ar@{.}[d]
        \ar@{-|>}@/_.1pc/[d]&&
\bq^+   \ar@{.}[d]
        \ar@{-|>}@/_.1pc/[d]\\
\bq^-   \ar@{.}[d]
        \ar@{-|>}[urrr]&&
\bq^-   \ar@{.}[d]
        \ar@{-|>}[urrr]&
\bq^-   \ar@{-|>}[dlll]&&
\bq^-   \ar@{-|>}[dlll]\\
\bq^+&&\bq^+
}$}$}
\raisebox{20pt}{$
\scalebox{.75}{$
\xymatrix@C=-4pt@R=5pt{
&&&\bq^-\ar@{.}[dll]
        \ar@{.}[d]
        \ar@{.}[drr]
	\ar@{-|>}@/_.1pc/[drr]
\\
&\bq^+  \ar@{.}[dl]
        \ar@{.}[dr]
	\ar@{-|>}@/_.1pc/[dl]
	\ar@{-|>}@/_.1pc/[dr]
&&
\bq^+   \ar@{.}[d]
	\ar@{-|>}@/_.1pc/[d]
&&
\bq^+   \ar@{.}[d]
	\ar@{-|>}@/_.1pc/[d]
\\
\bq^-   \ar@{.}[d]
	\ar@{-|>}@/_.1pc/[d]
&&
\bq^-   \ar@{.}[d]
	\ar@{-|>}@/_.1pc/[d]
&
\bq^-   \ar@{-|>}@/_.1pc/[ull]
&&
\bq^-   \ar@{-|>}@/_.1pc/[ull]
\\
\bq^+&&\bq^+
}$}$}
\]
\caption{Non-unique causal explanation}
\label{fig:non_uniq_caus}
\end{minipage}
\hfill
\begin{minipage}{.5\linewidth}
\[
\raisebox{20pt}{$
\scalebox{.75}{$
\xymatrix@C=-4pt@R=6pt{
&&&&&&\bq^-
	\ar@{.}[dlll]
	\ar@{.}[d]
	\ar@{.}[drr]\\
&&&\bq^+
	\ar@{.}[dl]
	\ar@{.}[dr]
&&&\bq^+
	\ar@{.}[dl]
	\ar@{.}[d]
	\ar@{.}[dr]
&&\bq^+
	\ar@{.}[d]\\
&&\bq^-
	\ar@{.}[dl]
	\ar@{.}[d]
	\ar@{.}[dr]
&&\bq^-
	\ar@{.}[d]
&\bq^-
&\bq^-
&\bq^-
&\bq^-\\
&\bq^+&\bq^+&\bq^+&\bq^+
}$}$}
\leadsto\hspace{-15pt}
\raisebox{20pt}{$
\scalebox{.75}{$
\xymatrix@C=-4pt@R=6pt{
&&&&&&\bq^-
        \ar@{.}[dlll]
        \ar@{.}[d]
        \ar@{.}[drr]
	\ar@{-|>}@/_.1pc/[dlll]\\
&&&\bq^+
        \ar@{.}[dl]
        \ar@{.}[dr]
	\ar@{-|>}@/_.1pc/[dl]
	\ar@{-|>}@/_.1pc/[dr]
&&&\bq^+
        \ar@{.}[dl]
        \ar@{.}[d]
        \ar@{.}[dr]
	\ar@{-|>}@/_.1pc/[dl]
	\ar@{-|>}@/_.1pc/[d]
	\ar@{-|>}@/_.1pc/[dr]
&&\bq^+
        \ar@{.}[d]
	\ar@{-|>}@/_.1pc/[d]\\
&&\bq^-
        \ar@{.}[dl]
        \ar@{.}[d]
        \ar@{.}[dr]
	\ar@{-|>}@[red]@/_.1pc/[urrrr]
&&\bq^-
        \ar@{.}[d]
	\ar@{-|>}@[blue]@/_.1pc/[urrrr]
&\bq^-
	\ar@{-|>}[dllll]
&\bq^-
	\ar@{-|>}[dllll]
&\bq^-
	\ar@{-|>}[dllll]
&\bq^-
	\ar@{-|>}[dllll]	
\\
&\bq^+&\bq^+&\bq^+&\bq^+
}$}$}
\]
\caption{Unique causal explanation}
\label{fig:pos_char_kx}
\end{minipage}
\end{figure}
Indeed, in Figure \ref{fig:non_uniq_caus}, we show on the left hand side the
configuration $z_1$ underlying Figure \ref{fig:caus_kx_exp} -- up to
iso it has exactly two causal explanations, shown on the right.
The rightmost augmentation is not an expansion of $K_x$, so $K_x$ is not
the only strategy featuring (the isomorphism class of) $z_1$. 
However, we \emph{can} find a position unique to $K_x$. Consider $z_2$
the configuration on the left hand side of Figure \ref{fig:pos_char_kx}.
The only possible augmentation (up to iso) yielding $z_2$ as a
desequentialization appears on the right hand side (call it $\q$): every
other attempt to guess causal wiring fails. In particular, the red and
blue immediate causal links are forced by the cardinality of the subsequent
duplications. But $\q$ is an expansion of the unique maximal branch of
$K_x$ -- so it suffices to see $z_2$ in $\positions{\sigma}$ to know that
$\sigma = K_x$.

This suggests a proof idea: given $\p_1, \p_2 : A$ causal
strategies with $\positions{\p_1} = \positions{\p_2}$, we devise a
\emph{characteristic expansion} of $\p_1$ with 
duplications chosen to make the causal structure essentially unique;
meaning it must be an expansion of $\p_2$ as well. We do this by using:

\begin{definition}
A \textbf{fork} in $\q \in \Aug(A)$ is a maximal non-empty set $X
\subseteq \ev{\q}$ s.t.:
\[
\begin{array}{rl}
\text{\emph{negative:}} & \text{for all $a \in X$, $\lambda(a) = -$,}\\
\text{\emph{sibling:}} & \text{$X = \{\init(\q)\}$ or there is $b \in
\ev{\q}$ such that for all $a \in X$, $b \imc_\q a$,}\\
\text{\emph{identical:}} & \text{for all $a_1, a_2 \in X$, $\partial_\q(a_1)
= \partial_\q(a_2)$.}
\end{array}
\]

We write $\twin(\q)$ for the set of forks in augmentation $\q$.
\end{definition}

If $\p$ is a causal strategy, $\q \in \exp(\p)$ and $X \in \twin(\q)$,
the definition of expansions ensures that all Player moves caused by
Opponent moves in $X$ are copies. So if $X$ has \textbf{cardinality}
$\sharp X = n$, and if we find exactly one set of cardinality $\geq n$ of
equivalent Player moves in $\deseq{\q}$, we may deduce that there is a
causal link. For instance, in Figure \ref{fig:pos_char_kx}, the causal
successors for the fork of cardinality $3$ may be found so. 
In general though, several Opponent moves may cause indistinguishable Player
moves, so that the cardinality of a set $Y$ of duplicated Player moves is
the sum of the cardinalities of the predecessor forks. To allow us to
identify these predecessor sets uniquely, the trick is to construct
the expansion so that all forks have cardinality a distinct power of
$2$, making it so that the predecessor forks can be inferred from
the binary decomposition of $\sharp Y$. This brings us to the following
definition.

\begin{definition}
A \textbf{characteristic expansion} of $\p$ is $\q \in \exp(\p)$ such that:
\[
\begin{array}{rl}
\text{\emph{injective:}} & \text{for $X, Y \in \twin(\q)$, if $\sharp X
= \sharp Y$ then $X = Y$,}\\
\text{\emph{well-powered:}} & \text{for all $X \in \twin(\q)$, there is
$n \in \mathbb{N}$ such that $\sharp X = 2^n$,}\\
\text{\emph{$-$-obsessional:}} & \text{for all $a^+ \in \ev{\q}$, if
$\partial_\q(a^+) \imc_A b^-$, there is $a^+ \imc_\q a'$ s.t.
$\partial_\q(a') = b^-$.}
\end{array}
\]
\end{definition}

This only constrains causal links in $\q$
from positives to negatives, but by \emph{courteous} those are in
$\q$ iff they are in $\deseq{\q}$. So for $\q \in \exp(\p)$,
that it is a characteristic expansion is in fact a property of
$\deseq{\q}$. Furthermore it is stable under iso so that if
$\positions{\p_1} = \positions{\p_2}$, for $\q_1 \in \exp(\p_1)$ 
characteristic there must be $\q_2 \in \exp(\p_2)$ 
characteristic too such that $\deseq{\q_1} \iso \deseq{\q_2}$ -- so it
makes sense to restrict our attention to positions reached by
characteristic expansions.

How different can be characteristic $\q_1 \in \exp(\p_1)$ and
$\q_2 \in \exp(\p_2)$ s.t. $\deseq{\q_1} \iso \deseq{\q_2}$? A
first guess is \emph{isomorphic}, but that is off the mark;
$\q_1$ and $\q_2$ have some degree of liberty in swapping
forks around (as in Figure \ref{fig:ex_bisim}): they have the
``same branches, but with possibly different multiplicity''. A
significant part of our endeavour has been to construct a relation
between augmentations allowing such changes in multiplicity, while
ensuring $\p_1 \iso \p_2$. 

\subsection{Bisimulations Across an Isomorphism}
\begin{figure}
\[
\scalebox{.8}{$
\xymatrix@R=3pt@C=4pt{
&&&\qu^-
	\ar@{-|>}[d]
&&&&&\sim^\varphi&&&&&
\qu^-	\ar@{-|>}[d]\\
&&&\qu^+\ar@{.}@/^/[u]
	\ar@{-|>}[dll]
	\ar@{-|>}[drr]
&&&&&&&&&&
\qu^+	\ar@{.}@/^/[u]
	\ar@{-|>}[dll]
	\ar@{-|>}[drr]\\
&\qu_1^-\ar@{.}@/^/[urr]
	\ar@{-|>}[d]&&&&
\qu_2^-	\ar@{.}@/_/[ull]
	\ar@{-|>}[d]&&&&&&
\qu_1^-	\ar@{.}@/^/[urr]
	\ar@{-|>}[d]&&&&
\qu_2^-	\ar@{.}@/_/[ull]
	\ar@{-|>}[d]\\
&\cyan{\qu^+}
	\ar@{.}@/^2pc/[uuurr]
	\ar@{-|>}[dl]
	\ar@{-|>}[dr]&&&&
\qu^+	\ar@{.}@/_2pc/[uuull]
	\ar@{-|>}[d]&&&
\cyan{\sim^\varphi}&&&
\cyan{\qu^+}
	\ar@{.}@/^2pc/[uuurr]
	\ar@{-|>}[d]&&&&
\qu^+	\ar@{.}@/_2pc/[uuull]
	\ar@{-|>}[dl]
	\ar@{-|>}[dr]\\
\green{\qu^-}
	\ar@{.}@/^/[ur]
	\ar@{-|>}[d]&&
\qu^-	\ar@{.}@/_/[ul]
	\ar@{-|>}[d]&&&
\qu^-	\ar@{.}@/_/[u]
	\ar@{-|>}[d]&&&&&&
\green{\qu^-}
	\ar@{.}@/^/[u]
	\ar@{-|>}[d]&&&
\qu^-	\ar@{.}@/^/[ur]
	\ar@{-|>}[d]&&
\qu^-	\ar@{.}@/_/[ul]
	\ar@{-|>}[d]\\
\red{\qu^+}
	\ar@{.}@/^/[u]&&
\qu^+	\ar@{.}@/_/[u]&&&
\qu^+
	\ar@{.}@/_/[u]&&&
\red{\sim}^{\red{\varphi}}_{\red{\{(}\green{\qu^-}\red{,}
\green{\qu^-}\red{)\}}}&&&
\red{\qu^+}
	\ar@{.}@/^/[u]&&&
\qu^+	\ar@{.}@/^/[u]&&
\qu^+	\ar@{.}@/_/[u]
}
$}
\]
\caption{Distinct characteristic expansions reaching the same position}
\label{fig:ex_bisim}
\end{figure}

More than simply comparing augmentations, given $\q, \p \in \Aug(A)$, $a
\in \ev{\q}$, $b \in \ev{\p}$, we shall need a
a predicate $a \sim b$ expressing that $a$ and $b$ have the same
causal follow-up, up to the multiplicity of duplications. In particular,
$a$ and $b$ must have ``the same pointer'', but at first that makes no sense
since $a$ and $b$ live in different ambient sets of events. So we also
fix an isomorphism $\varphi : \deseq{\q} \iso \deseq{\p}$ providing the
translation, and aim to define $a \sim^\varphi b$ parametrized by
$\varphi$. We give some examples in Figure \ref{fig:ex_bisim}, where
$\varphi$ is any of the two possible isomorphisms, assuming $\qu_1^-$ and
$\qu_2^-$ correspond to different moves of the arena. 


This is defined via a bisimulation game: for instance,
establishing that the roots are in relation requires us to first
match the blue nodes. But as the bisimulation unfolds, requiring
all pointers to match up to $\varphi$ is too strong: the pointers of red
moves do \emph{not} match -- but seen from
$\cyan{\qu^+}$ this is fine as the justifiers for the red moves are
encountered at the same step of the bisimulation game from
$\cyan{\qu^+}$. So our actual predicate has form $a
\sim^\varphi_\Gamma b$ for $\Gamma$ a \emph{context}, stating
a correspondence between negative moves established in the
bisimulation game so far:

\begin{definition}
A \textbf{context} between $\q, \p \in \Aug(A)$ is
$\Gamma : \dom(\Gamma) \iso \cod(\Gamma)$
a bijection s.t. $\dom(\Gamma) \subseteq \ev{\q}$, $\cod(\Gamma)
\subseteq \ev{\p}$, $\lambda_\q(\dom(\Gamma)) \subseteq \{-\}$, and
$\forall a^- \in \dom(\Gamma)$, $\partial_\q(a) =
\partial_\p(\Gamma(a))$. 
\end{definition}

We may now formulate a first notion of bisimulation across
augmentations.

\begin{definition}\label{def:main_bisim}
Consider $\q, \p \in \Aug(A)$ and an isomorphism $\varphi : \deseq{\q}
\iso \deseq{\p}$.

For $a \in \ev{\q}$, $b \in \ev{\p}$ and $\Gamma$ a context, we define a
predicate $a \sim_\Gamma^\varphi b$ which holds if, firstly, 
\[
\begin{array}{rcl}
\text{\emph{(a)}} && \text{$\partial_\q(a) = \partial_\p(b)$ and $\Gamma \vdash (a,b)$}\\
\text{\emph{(b)}} && 
\text{if $\just(a^+) \in \dom(\Gamma)$, then $\just(b) \in
\cod(\Gamma)$ and $\Gamma(\just(a)) = \just(b)$,}\\
\text{\emph{(c)}} &&
\text{if $\just(a^+) \not \in \dom(\Gamma)$, then
$\just(b) \not \in \cod(\Gamma)$ and $\varphi(\just(a)) = \just(b)$,}
\end{array}
\]
where $\Gamma \vdash (a, b)$ means that for all $a' \in \dom(\Gamma)$,
$\neg (a' >_\q a)$ and for all $b' \in \cod(\Gamma)$, $\neg (b' >_\p
b)$; and inductively, the following two bisimulation conditions hold:
\[
\begin{array}{rcl}
\text{\emph{(1)}} &&
\text{if $a^+ \imc_\q a'$, then there is $b^+ \imc_\p b'$ with $a'
\sim_{\Gamma \cup \{(a',b')\}}^\varphi b'$, and symmetrically,}\\
\text{\emph{(2)}} &&
\text{if $a^- \imc_\q a'$, then there is $b^- \imc_\p b'$ with $a'
\sim_\Gamma^\varphi b'$, and symmetrically.}
\end{array}
\]
\end{definition}

As $\Gamma \vdash (a, b)$ implies $a' \not \in \dom(\Gamma)$ and $b'
\not \in \cod(\Gamma)$, $\Gamma \cup \{(a', b')\}$ remains a bijection.

Of particular interest is the case $a
\sim_\emptyset^\varphi b$ over an empty context, written simply $a
\sim^\varphi b$. From this, we deduce a relation between augmentations:
we write  $\q \sim^\varphi \p$ iff $\init(\q) \sim^\varphi \init(\p)$,
for $\q, \p \in \Aug(A)$ and $\varphi : \deseq{\q} \iso \deseq{\p}$.
Resuming the discussion at the end of Section
\ref{subsec:char_exp}: bisimulations allow us to express that two
characteristic expansions with isomorphic configurations are ``the
same''. More precisely, in due course we will be able to prove:

\begin{proposition}\label{prop:main}
Consider $\p_1, \p_2 \in \Aug(A)$ causal strategies, $\q_1 \in
\exp(\p_1)$ and $\q_2 \in \exp(\p_2)$ characteristic expansions with an
isomorphism $\varphi : \deseq{\q_1} \iso \deseq{\q_2}$. 
Then, $\q_1 \sim^\varphi \q_2$.
\end{proposition}

The proof is the core of our injectivity argument, which we will cover
in Section \ref{subsec:pos_inj}. For now, we focus on how to
conclude from $\q_1 \sim^\varphi \q_2$ that we have $\p_1 \iso \p_2$.

\subsection{Compositional Properties of Bisimulations}\label{subsec:comp_bissim}

To achieve that, we exploit compositional properties of
bisimulations. More precisely, we show that $\q_i \in
\exp(\p_i)$ induces a bisimulation $\q_i \sim \p_i$, and 
find a way to compose
\begin{eqnarray}
\p_1 \sim \q_1 &\sim^\varphi& \q_2 \sim \p_2	\label{eq:1}
\end{eqnarray}
to deduce $\p_1 \sim \p_2$ in a sense yet to be defined, and
$\p_1 \iso \p_2$ will follow. We start with:

\begin{restatable}{lemma}{biseq}\label{lem:bis_eq}
Consider augmentations $\q, \p, \r \in \Aug(A)$, isomorphisms $\varphi :
\deseq{\q} \iso \deseq{\p}$, $\psi : \deseq{\p} \iso \deseq{\r}$, 
events $a \in \ev{\q}$, $b \in \ev{\p}$, $c \in \ev{\r}$, and contexts
$\Gamma, \Delta$. Then:
\[
\begin{array}{rl}
\text{\emph{reflexivity:}} &
\text{$a \sim^\id a$,}\\ 
\text{\emph{transitivity:}} &
\text{if $a \sim_\Gamma^\varphi b$ and $b \sim_\Delta^\psi c$ with
$\cod(\Gamma) = \dom(\Delta)$, then $a
\sim_{\Delta \circ \Gamma}^{\psi \circ \varphi} c$,}\\ 
\text{\emph{symmetry:}} &
\text{if $a \sim_\Gamma^\varphi b$ then $b
\sim_{\Gamma^{-1}}^{\varphi^{-1}} a$.}
\end{array}
\]
\end{restatable}

But in order to treat $\q_i \in \exp(\p_i)$ as a bisimulation between
$\q_i$ and $\p_i$, Definition \ref{def:main_bisim} does not do the
trick: we cannot expect there to be an iso between $\deseq{\q_i}$ and
$\deseq{\p_i}$ as $\q_i$ has by construction many more events.  We
therefore introduce a variant of Definition \ref{def:main_bisim}:

\begin{definition}\label{def:var_bisim}
Consider $\q, \p \in \Aug(A)$.
For $a \in \ev{\q}$, $b \in \ev{\p}$, $\Gamma$, we have
$a \sim_\Gamma b$ if
\[
\begin{array}{rcl}
\text{\emph{(a)}} && \text{$\partial_\q(a) = \partial_\p(b)$ and $\Gamma \vdash
(a,b)$,}\\
\text{\emph{(b)}} && 
\text{$\just(a^+) \in \dom(\Gamma)$ and $\Gamma(\just(a)) =
\just(b)$,}\\
\text{\emph{(1)}} &&
\text{if $a^+ \imc_\q a'$, then $b^+ \imc_\p b'$ with $a'
\sim_{\Gamma \cup \{(a',b')\}} b'$, and symmetrically,}\\
\text{\emph{(2)}} &&
\text{if $a^- \imc_\q a'$, then $b^- \imc_\p b'$ with $a'
\sim_\Gamma b'$, and symmetrically.}
\end{array}
\]
\end{definition}

This helps us relate $\q$ and $\p$ when $\deseq{\q}$
and $\deseq{\p}$ are not isomorphic: we set $\q \sim \p$ iff $\init(\q)
\sim_{\{(\init(\q), \init(\p))\}} \init(\p)$. A variation of
Lemma \ref{lem:bis_eq} shows $\sim$ is an equivalence, and:

\begin{restatable}{proposition}{simcharexp}\label{prop:sim_char_exp}
Consider $A$ an arena, $\p \in \Aug(A)$ a causal strategy, and $\q \in
\Aug(A)$.

Then, $\q$ is a $-$-obsessional expansion of $\p$ iff $\q \sim \p$.
\end{restatable}
\begin{proof}
\emph{If.} We simply construct $\varphi : \q \to \p$ for all $a\in \ev{\q}$ by
induction on $\leq_\q$. The image is provided by bisimulation, its
uniqueness by \emph{determinism} and \emph{$-$-linearity}.

\emph{Only if.}
For $\varphi : \q \to \p$ and $a \in \ev{\q}$, write $[a]^-_\q = \{a'
\in \ev{\q} \mid a'\leq_\q a~\&~\lambda(a') = -\}$; 
it is totally ordered by $\leq_\q$ as $\q$ is \emph{forestial}.
From the conditions on $\varphi$, it is direct that it induces an
order-iso $[a]^-_\q \iso [\varphi (a)]^-_\p$, \emph{i.e.}
a context $\Gamma\tuple{a} : [a]^-_\q \iso [\varphi (a)]^-_\p$. Then,
by Lemma \ref{lem:oexp_to_sim}, we obtain $a \sim_{\Gamma\tuple{a}}
\varphi(a)$ for all $a \in \ev{\q}$. We then apply this to $\init(\q)$. 
\end{proof} 

This vindicates Definition \ref{def:var_bisim}. But for (\ref{eq:1}), we
must compose two kinds of bisimulations, following Definitions
\ref{def:main_bisim} and \ref{def:var_bisim}. Fortunately, whenever both
definitions apply, they coincide:

%
%
%

\begin{restatable}{lemma}{removeiso}\label{lem:remove_iso}
Consider $\q, \p \in \Aug(A)$, and $\varphi : \deseq \q \iso \deseq \p$.
Then, $\q \sim^\varphi \p$ iff $\q \sim \p$.
\end{restatable}
\begin{proof}
\emph{If.} Straightforward from Definitions \ref{def:main_bisim}
and \ref{def:var_bisim}: case \emph{(c)} is never used.

\emph{Only if.} We actually prove that for all $a \in \ev{\q}$, $b \in
\ev{\p}$, for all context $\Gamma$ which is \emph{complete} in the sense
that $[a]_\q^-\subseteq\dom(\Gamma)$ and
$[b]_\p^-\subseteq\cod(\Gamma)$, if $a \sim_{\Gamma}^\varphi b$ then $a
\sim_{\Gamma} b$. The proof is immediate by induction: the clause
\emph{(c)} is never used from the hypothesis that $\Gamma$ is complete.
Finally, we apply this to the roots of $\q, \p$ with context
$\{(\init(\q), \init(\p))\}$.
\end{proof}

Altogether, we have: 

\begin{proposition}\label{prop:bisim_to_iso}
Consider $\p_1, \p_2 \in \Aug(A)$ causal strategies, $\q_1 \in
\exp(\p_1), \q_2 \in \exp(\p_2)$ characteristic expansions with an
iso $\varphi : \deseq{\q_1} \iso \deseq{\q_2}$.
If $\q_1 \sim^\varphi \q_2$, then $\p_1 \iso \p_2$.
\end{proposition}
\begin{proof}
By Lemma \ref{lem:remove_iso}, $\q_1 \sim \q_2$.
As characteristic expansions, $\q_1$ and $\q_2$ are $-$-obsessional, so
by Proposition \ref{prop:sim_char_exp}, $\q_1 \sim \p_1$ and
$\q_2 \sim \p_2$. So 
$\p_1 \sim \q_1 \sim \q_2 \sim \p_2$
but $\sim$ is an equivalence, so $\p_1 \sim \p_2$.
By Proposition \ref{prop:sim_char_exp}, we have
$\varphi : \p_1 \to \p_2$ and $\psi : \p_2 \to \p_1$
composing to $\psi \circ \varphi : \p_1 \to \p_1$. But by Lemma
\ref{lem:mor_unique},  there is only one morphism from $\p_1$ to itself,
the identity, so $\psi \circ \varphi = \id$. Likewise $\varphi \circ
\psi = \id$, hence $\varphi : \p_1 \iso \p_2$ as required.
\end{proof}

\subsection{Clones}\label{subsec:clones}

In Section \ref{subsec:char_exp}, we introduced \emph{characteristic
expansions} which, via duplications with well-chosen cardinalities,
constrain the causal structure. More precisely, if $\q \in \exp(\p)$ is
characteristic, looking at a set of duplicated Player moves in $\deseq{\q}$
of cardinality $n$ as in Figure \ref{fig:copy_player}, decomposing
$n = \sum_{i\in I} 2^i$,
we can deduce that the causal predecessors of the $\bq^+_j$'s are among
the forks with cardinality $2^i$ for $i\in I$. But that is not enough: this
does not tell us how to distribute the $\bq^+_j$'s to the forks,
and not all the choices will work: while the $\bq^+_j$'s are copies,
their respective causal follow-ups might differ. So the idea is simple:
imagine that the causal follow-ups for the $\bq^+_j$'s are already
reconstructed. Then we may compare them using bisimulation, and
replicate the same reasoning as above on bisimulation equivalence classes. 

So we are left with the task of leveraging bisimulation to define an
adequate equivalence relation on $\ev{\q}$. This leads to the notion of
\emph{clones}, our last technical tool.

\begin{definition}\label{def:clone}
Consider $\q, \p \in \Aug(A)$, $\varphi : \deseq{\q} \iso \deseq{\p}$,
and $a \in \ev{\q}$, $b \in \ev{\p}$. 

We say that $a$ and $b$ are \textbf{clones} through $\varphi$, written
$a \clone^\varphi b$, if there is a context $\Gamma$ preserving
pointers (\emph{i.e.} for all $a' \in \dom(\Gamma)$, $\varphi(\just(a'))
= \just(\Gamma(a'))$) such that $a \sim^\varphi_\Gamma b$.
\end{definition}

This allows $a$ and $b$ (and their follow-ups) to
change their pointers through some unspecified $\Gamma$.
\begin{figure}[t]
\begin{minipage}{.45\linewidth}
\[
\scalebox{.8}{$
\xymatrix@R=9pt@C=10pt{
&\bq^-
        \ar@{.}[dl]
        \ar@{.}[d]
        \ar@{.}[drr]\\
\bq^+_1&
\bq^+_2&\dots&
\bq^+_n\\~
}$}
\]
\caption{A set of copied Player moves}
\label{fig:copy_player}
\end{minipage}
\hfill
\begin{minipage}{.5\linewidth}
\[
\scalebox{.8}{$
\xymatrix@R=5pt@C=10pt{
&\bq^+
	\ar@{.}[dl]
	\ar@{-|>}@/_.1pc/[dl]
	\ar@{.}[d]
	\ar@{-|>}@/_.1pc/[d]
	\ar@{.}[drr]
	\ar@{-|>}@/_.1pc/[drr]\\
\bq_1^-	\ar@{.}[d]
	\ar@{-|>}@/_.1pc/[d]&
\bq_2^-	\ar@{.}[d]
	\ar@{-|>}@/_.1pc/[d]&\dots&
\bq_n^-	\ar@{.}[d]
	\ar@{-|>}@/_.1pc/[d]\\
\bq_1^+	\ar@{}[r]|\clone & 
\bq_2^+ \ar@{}[r]|\clone &
\dots	\ar@{}[r]|\clone & \bq_n^+
}$}
\]
\caption{A set of clones switching pointers}
\label{fig:fig_ex_clones}
\end{minipage}
\end{figure}
Indeed, the picture painted by Figure \ref{fig:copy_player} is
limited: a fork might trigger Player moves with different
pointers, as in Figure \ref{fig:fig_ex_clones}. As
$a \clone^\varphi b$ quantifies existentially over 
contexts, compositional properties of clones are 
more challenging. Nevertheless, via a canonical form
for contexts and leveraging Lemma \ref{lem:bis_eq}, we show that $a
\clone^{\id} a$, that $a \clone^\varphi b$ and $b \clone^\psi c$ imply
$a \clone^{\psi\circ \varphi} c$, and that $a \clone^\varphi b$ implies
$b \clone^{\varphi^-1} a$ whenever these typecheck -- see Appendix
\ref{app:clones}. Instantiating Definition
\ref{def:clone} with $\q = \p$ and $\varphi = \id$, we get an
equivalence relation $\clone$ on $\ev{\q}$.

Moreover, we have the crucial property that 
forks generate clones (see Appendix \ref{app:clones}):

\begin{restatable}{lemma}{twinclones}\label{lem:twin_clones}
Consider $\q$ a $-$-obsessional expansion of causal strategy $\p$ on
arena $A$.

Then, for all $a_1^-, a_2^- \in X \in \twin(\q)$, for all $a_1^- \imc_\q
b_1^+$ and $a_2^- \imc_\q b_2^+$, $b_1 \clone b_2$.
\end{restatable}

By Lemma \ref{lem:twin_clones}, if a clone class
includes a positive move, it also has all its cousins triggered by the
same fork -- so clone classes may be partitioned
following forks:

\begin{restatable}{lemma}{partition}\label{lem:partition}
Let $\q$ be a characteristic expansion of causal strategy $\p$, and $Y$
a clone class of positive events in $\ev{\q}$, with $\sharp Y =
\sum_{i\in I} 2^i$ for $I \subseteq \mathbb{N}$ finite.
Then, for all $i \in \mathbb{N}$, $i\in I$ iff there is $X_i \in \twin(\q)$ with $\sharp X_i =
2^i$ and $a^- \in X_i$, $b^+ \in Y$ such that $a^- \imc_\q b^+$.
\end{restatable}
\begin{proof}
For any $i\in\N$, we write $X_i$ the fork of $\q$ of cardinality $2^i$,
if it exists. 

Consider the set
$J:= \{j \in \N \tq X_j \text{ exists,} \; \exists a\in
X_j, \; \exists b\in Y, \; a\imc_{\q}b\}$.
Any $b\in Y$ is positive and so the unique (by determinism) successor of
some negative event $a$. Moreover $a$ appears in a fork $X$ and
by Lemma \ref{lem:twin_clones}, all events of $X$ are predecessors of
events of $Y$.
Hence, we have $Y = \bigcup_{j \in J} \suc(X_j)$,
        where the union is disjoint since $\q$ is forest-shaped.
Therefore,
        \[
        \sharp Y \;=\; \sum_{j \in J} \sharp\suc(X_j)
        \;=\; \sum_{j \in J} \sharp X_j
        \;=\; \sum_{j \in J} 2^j,
        \]
        where the second equality is obtained by determinism.
        By uniqueness of the binary decomposition, $J=I$, which proves
the lemma by definition of $J$.
\end{proof}

\subsection{Positional Injectivity}
\label{subsec:pos_inj}

We are finally in a position to prove the core of the injectivity
argument.

\begin{lemma}[Key lemma]\label{lem:key}
Consider $\p_1, \p_2 \in \Aug(A)$ causal strategies, $\q_1 \in
\exp(\p_1)$ and $\q_2 \in \exp(\p_2)$ characteristic expansions,
and $\varphi : \deseq{\q_1} \iso \deseq{\q_2}$.
Then, $\forall a^+ \in \ev{\q_1}, a \clone^\varphi \varphi(a)$.
\end{lemma}
\begin{proof}
The \textbf{co-depth} of $a \in \ev{\q_i}$ is the
maximal length $k$ of
$a = a_1 \imc_{\q_i} \dots \imc_{\q_i} a_k$
a causal chain in $\q_i$. We show by induction on $k$ the
two symmetric properties:
\[
\begin{array}{rl}
\text{\emph{(a)}}& 
	\text{for all $a^+ \in \ev{\q_1}$ of co-depth $\leq k$, we
have $a \clone^\varphi \varphi(a)$,}\\
\text{\emph{(b)}}&
	\text{for all $a^+ \in
\ev{\q_2}$ of co-depth $\leq k$, we have $a
\clone^{\varphi^{-1}} \varphi^{-1}(a)$.}
\end{array}
\] 

Take $a^+ \in \ev{\q_1}$ of co-depth $k$. If $a$ is maximal
in $\q_1$, so is $\varphi(a)$ in $\q_2$ and $a\clone \varphi(a)$. Else,
the successors of $a$ partition as
$G_1, \dots, G_n \subseteq \twin(\q_1)$, where
$G_i = \{b_{i,1}^-, \dots, b_{i, 2^{p_i}}^-\}$;
likewise the successors of $\varphi(a)$ in $\q_2$ are the
forks $\varphi(G_i)$. 
For all $1\leq i \leq n$ and $1 \leq j \leq 2^{p_i}$, we claim:
\begin{eqnarray}
&&\text{for all $b_{i,j} \imc_{\q_1} c_{i,j}$, there is
$\varphi(b_{i,j}) \imc_{\q_2} d_{i,j}$ satisfying $c_{i,j} \clone^\varphi
d_{i,j}$.} \label{eq:2}
\end{eqnarray}

Write $X = [c_{i,j}]_{\clone}$ the clone class of $c_{i,j}$ in
$\q_1$. It is easy to prove that the clone relation preserves co-depth,
so it follows
from the induction hypothesis and 
Lemma \ref{lem:clone_eq} that
$\varphi(X)$ is a clone class in $\q_2$. By Lemma \ref{lem:partition}, 
$\sharp X$ has $2^{p_i}$ in its binary decomposition -- and as $\varphi$
is a bijection, so does $\sharp (\varphi(X))$.
So by Lemma \ref{lem:partition}, there is $\varphi(b_{i,j}) \in
\varphi(G_i)$ and $d_{i,j} \in \varphi(X)$ such that $\varphi(b_{i,j}) \imc_{\q_2}
d_{i,j}$. Since $\varphi(c_{i,j}), d_{i,j} \in
\varphi(X)$ they are clones, so using
$c_{i,j} \clone^{\varphi} \varphi(c_{i,j})$ by induction hypothesis, 
$c_{i,j} \clone^\varphi d_{i,j}$.
Likewise, the mirror property of \eqref{eq:2} also holds. 

Deducing $a \clone^{\varphi} \varphi(a)$ requires some care: cloning is
defined via a context, and the $c_{i,j} \clone^\varphi \varphi(c_{i,j})$
might not share the same.
However, the contexts can be put into canonical forms that are shown to
agree -- Lemma \ref{lem:clone_lift} allows us to prove $a \clone^{\varphi}
\varphi(a)$ from \eqref{eq:2} and its mirror property.
Finally, \emph{(b)} is proved symmetrically.
\end{proof}

Now, consider $\p_1, \p_2, \q_1, \q_2, \varphi$ as in 
Proposition \ref{prop:main}. If the $\q_i$'s are empty or singleton
trees, there is nothing to prove. Otherwise $\q_i$ starts with $a_i^-
\imc_{\q_i} b_i^+$ with $a_i^-$ initial. But then $[b_i^+]_{\clone}$ is
the only singleton clone class in $\q_i$. As $\varphi$ preserves clone
classes, $\varphi(b_1^+) = b_2^+$. By Lemma \ref{lem:key}, $b_1
\clone^\varphi b_2$. Thus $b_1 \sim^\varphi b_2$, so $a_1
\sim^\varphi a_2$ and $\q_1 \sim^\varphi \q_2$. This
concludes the proof of Proposition \ref{prop:main}.
Putting everything together, we obtain:

\begin{theorem}\label{th:main_caus_inj}
For $\p_1, \p_2 \in \Aug(A)$ causal strategies s.t. 
$\positions{\p_1} = \positions{\p_2}$, then $\p_1 \iso \p_2$.
\end{theorem}
\begin{proof}
Consider $\q_1 \in \exp(\p_1)$ a characteristic expansion. By
hypothesis, there must be $\q_2 \in \exp(\p_2)$ and $\varphi :
\deseq{\q_1} \iso \deseq{\q_2}$; necessarily $\q_2$ is also a
characteristic expansion of $\p_2$. By Proposition \ref{prop:main}, we
have $\q_1 \sim^\varphi \q_2$. By Proposition \ref{prop:bisim_to_iso},
we have $\p_1 \iso \p_2$.
\end{proof}
%

Finally, Theorem \ref{th:main} follows from Theorem
\ref{th:main_caus_inj}, Proposition \ref{prop:pos_sig_pos_caus} and
Lemma \ref{lem:strat_eq_caus_iso}.

Theorem \ref{th:main} only concerns \emph{total} finite innocent
strategies. In contrast, Theorem \ref{th:main_caus_inj} requires no
totality assumption: totality comes in not in the injectivity
argument, but in Proposition \ref{prop:pos_sig_pos_caus} linking
standard and causal strategies.
Without totality, expansions of
$\caus{\sigma}$ might not have as many Opponent as Player moves, and so
may not be linearizable via alternating plays. Intuitively,
in alternating plays Opponent may only explore
converging parts of the strategy, whereas in the causal setting Opponent
is free to explore simultaneously many branches, including divergences.
Positional injectivity for \emph{partial} finite innocent
strategies may be studied causally by restricting to \emph{$+$-covered}
expansions, \emph{i.e.} with only Player maximal events. But
then we must also abandon \emph{$-$-obsessionality} as Opponent moves
leading to divergence will not be played, breaking our proof
(Lemma \ref{lem:twin_clones} fails) in a way for which we see no fix.

\section{Beyond Total Finite Strategies}
\label{sec:beyond_total_finite}

Finally, we show some subtleties and partial
results on generalizations of Theorem \ref{th:main}.

First, positional injectivity fails in general.
Consider the infinitary terms $f : o \to o \to o \vdash
T_1, T_2, L, R : o$ recursively defined as
$T_1 = f\,T_2\,R, T_2 = f\,L\,T_1, L = f\,L\,\bot$ and $R = f\,\bot\,R$
in an infinitary simply-typed $\lambda$-calculus with divergence $\bot$.
The corresponding strategies differ:
\begin{figure}[t]
\begin{minipage}{.3\linewidth}
\[
\scalebox{.7}{$
\xymatrix@C=20pt@R=5pt{
&\red{\bq^-}  
	\ar@[red]@{.}[d]
	\ar@[red]@{-|>}@/_.1pc/[d]
	\ar@[red]@{.}@/_1pc/[ddd]
	\ar@{.}@/_3pc/[dddddl]
	\ar@{.}@/^2pc/[dddr]\\
&\red{\bq^+}
	\ar@[red]@{.}[d]
	\ar@{.}[dr]
	\ar@[red]@{-|>}@/_.1pc/[d]
	\ar@{-|>}@/_.1pc/[dr]\\
&\red{\bq^-_1}
	\ar@[red]@{-|>}@/_.1pc/[d]
&\bq^-_2\ar@{-|>}@/_.1pc/[d]\\
&\red{\bq^+}
	\ar@{.}[dl]
	\ar@[red]@{.}[d]
	\ar@{-|>}@/_.1pc/[dl]
	\ar@[red]@{-|>}@/_.1pc/[d]
&\bq^+	\ar@{.}[d]
	\ar@{-|>}@/_.1pc/[d]\\
\bq^-_1	\ar@{-|>}@/_.1pc/[d]
&\red{\bq^-_2}
	\ar@[red]@{-|>}@/^3pc/[uuu]
&\bq^-_2\ar@{-|>}@/_1.5pc/[u]\\
\bq^+	\ar@{.}[d]
	\ar@{-|>}@/_.1pc/[d]\\
\bq^-_1	\ar@{-|>}@/^1.5pc/[u]
}$}
\]
\caption{$\intr{\lambda f^{o\to o \to o}.\,T_1}$}
\label{fig:t1}
\end{minipage}
\hfill
\begin{minipage}{.3\linewidth}
\[
\scalebox{.7}{$
\xymatrix@C=20pt@R=5pt{
&\red{\bq^-}
	\ar@[red]@{-|>}@/_.1pc/[d]
	\ar@[red]@{.}[d]
	\ar@{.}@/_2pc/[dddl]
	\ar@[red]@{.}@/^1pc/[ddd]
	\ar@{.}@/^3pc/[dddddr]\\
&\red{\bq^+}
	\ar@{.}[dl]
	\ar@[red]@{.}[d]
	\ar@{-|>}@/_.1pc/[dl]
	\ar@[red]@{-|>}@/_.1pc/[d]\\
\bq^-_1	\ar@{-|>}[d]
&\red{\bq^-_2}
	\ar@[red]@{-|>}[d]\\
\bq^+	\ar@{.}[d]
	\ar@{-|>}@/_.1pc/[d]
&\red{\bq^+}
	\ar@[red]@{.}[d]
	\ar@{.}[dr]
	\ar@[red]@{-|>}@/_.1pc/[d]
	\ar@{-|>}@/_.1pc/[dr]\\
\bq^-_1	\ar@{-|>}@/^1.5pc/[u]
&\red{\bq^-_1}
	\ar@[red]@{-|>}@/_3pc/[uuu]
&\bq^-_2\ar@{-|>}[d]\\
&&\bq^+	\ar@{.}[d]
	\ar@{-|>}@/_.1pc/[d]\\
&&\bq^-_2
	\ar@{-|>}@/_1.5pc/[u]
}$}
\]
\caption{$\intr{\lambda f^{o \to o \to o}.\,T_2}$}
\label{fig:t2}
\end{minipage}
\hfill
\begin{minipage}{.3\linewidth}
\[
\xymatrix@R=12pt@C=-5pt{
~\\
&&&\bq^+
        \ar@{.}[dlll]
        \ar@{.}[dl]
        \ar@{.}[dr]
        \ar@{.}[drrr]\\
\bq_1^-& \dots& \bq_1^-&&
\bq_2^-& \dots& \bq_2^-\\~\\~
}
\]
\caption{Bricks}
\label{fig:bricks}
%
\end{minipage}
\end{figure}
their causal representations appear in Figures
\ref{fig:t1} and \ref{fig:t2}, infinite trees represented via loops.

We consider positions reached by plays -- or equivalently, by 
$+$-covered
expansions 
of Figures \ref{fig:t1} and
\ref{fig:t2}. In fact, both strategies admit
all \emph{balanced} positions on $\intr{(o \to o \to o) \to o}$, \emph{i.e.}
with as many Opponent as Player moves. Ignoring the
initial $\bq^-$, a position is a multiset of \textbf{bricks} as in
Figure \ref{fig:bricks},
with $i \in \mathbb{N}$ occurrences of $\bq_1^-$ and $j \in \mathbb{N}$
of $\bq_2^-$. A brick with $i = j = 0$ is a \textbf{leaf}.
The position is balanced if it has as many
Opponent as Player moves.  

Now, any position can be realized in $\intr{\lambda f^{o \to o \to
o}.\,T_i}$ by first placing bricks with occurrences of both $\bq^-_1$
and $\bq^-_2$ greedily alongside the
\emph{spine}, shown in red in Figures \ref{fig:t1} and \ref{fig:t2}. At
each step, we continue from only one of the copies opened, leaving 
others dangling.  If this gets stuck, apart from leaves we are left with
only $\bq^-_1$'s, or, only $\bq^-_2$'s, but there is always a
matching non-spine infinite branch available.  Finally, leaves can
always be placed as their number matches that of trailing negative moves
by the balanced hypothesis.

We have $\positions{\intr{\lambda f^{o \to o \to o}.\,T_1}} =
\positions{\intr{\lambda f^{o\to o \to o}.\,T_2}}$ as both strategies can
realize \emph{all} balanced positions on the arena $\intr{o \to o \to
o}$, and \emph{exactly} those: positional injectivity fails.

Positionality for \emph{finite} innocent strategies remains open.
We could only prove:

\begin{theorem}\label{th:partial}
Let $\sigma_1, \sigma_2 : A$ be finite innocent strategies with
$\positions{\sigma_1} = \positions{\sigma_2}$.

Then, $\sigma_1$ and $\sigma_2$ have the same P-views of maximal length.
\end{theorem}

For the proof (see Appendix \ref{app:partial}), we assume $\sigma_1$ has a
P-view $s$ of maximal length $n$. We perform an expansion of
$s$ where each Opponent branching at co-depth $2d+1$ has arity $d+1$. By a combinatorial
argument on trees, the only way to reassemble its nodes exhaustively 
in a tree with depth bounded by $d$ is to rebuild exactly the same tree.
Hence the tree is also in $\exp(\caus{\sigma_2})$, and $s \in \sigma_2$.
This steers us into conjecturing that positional injectivity holds for
partial finite innocent strategies, but our proof attempts have
remained inconclusive.

\section{Conclusion}
\label{sec:conclusion}

Though innocent strategies in
the Hyland-Ong sense are not positional, total finite innocent
strategies satisfy \emph{positional injectivity} -- however, the property
fails in general.

Beyond its foundational value, we believe this result may be 
helpful in the game semantics toolbox. Game semantics can be
fiddly; in particular, proofs that two terms yield the same
strategy are challenging to write in a concise yet rigorous manner.
This owes a lot to the complexity of \emph{composition}: proving that a
play $s$ is in $\intr{M\,N}$
involves constructing an ``interaction witness'' obtained from plays
in $\intr{M}$ and $\intr{N}$ plus an adequate ``zipping'' of the two.
Manipulations of plays with pointers are tricky and error-prone,
and the link between plays and terms is obfuscated by the multi-layered
interpretation.

In contrast, Theorem \ref{th:main} lets us prove innocent strategies
equal by comparing their positions. Now, constructing
a position of $\intr{M N}$ simply involves exhibiting matching
positions for $\intr{M}$ and $\intr{N}$. Side-stepping the
interpretation, this can be presented as typing terms with positions or
configurations -- combining Section \ref{subsec:link_rel} and
the link between relational semantics and
non-idempotent intersection type systems \cite{Carvalho18}.  For
instance, in this way, finite definability, a basic result seldom
presented in full formal details, boils down to typing the defined term
with the same positions as the original strategy.



\bibliography{biblio}

\newpage
\appendix

\section{Augmentations and Strategies: Proofs from Section \ref{sec:caus}}\label{app:caus}
In the sequel, $A$ is a fixed arena.  For any augmentation $\q\in
\Aug(A)$ and event $a\in\ev{\q}$, we define $\suc(a):= \{b \tq a \imc_\q
b\}$ the set of immediate successors of $a$ in $\leq_\q$. We also define
$\up a$ the set of descendants of $a$, i.e. $\up a := \{a'\tq a\leq_\q
a'\}$. 

In this first section, we provide more detailed proofs for augmentations as defined in Section \ref{sec:caus}, studying their expansions and their links with strategies. 
We start by proving Proposition \ref{prop:pview_aug} of Section
\ref{subsec:strat-to-caus}, which allows us to see traditional innocent
strategies as causal strategies, with the following two lemmas. For any
non-empty pointing string $s$, we use the notation $s_\omega$ for the
last element of $s$.

\begin{lemma}\label{lem:pview_aug}
	For $\sigma : A$ finite innocent on $A$ well-opened, ${\caus{\sigma}}$ is a causal strategy, \emph{i.e.} a receptive $-$-linear augmentation.
\end{lemma}
\begin{proof}	
	First we must check that $\deseq{\caus{\sigma}}=\tuple{\ev{{\caus{\sigma}}}, \leq_{\deseq{{\caus{\sigma}}}}, \partial_{\caus{\sigma}}}$ is a configuration. It is clear from the definition and the fact that $A$ is well-opened that $\tuple{\ev{{\caus{\sigma}}}, \leq_{\deseq{{\caus{\sigma}}}}}$ is a finite tree. The well-openedness of $A$ also ensures \emph{minimality-respecting.}
	
	\emph{Causality-preserving.} Consider $s, t
\in\ev{\caus{\sigma}}$ such that $s\imc_{\deseq{\caus{\sigma}}} t$. Then $t_\omega$ points to $s_\omega$ by definition, so by rigidity $\partial_{\caus{\sigma}}(s)\imc_A\partial_{\caus{\sigma}}(t)$.
	
	Hence, $\deseq{\caus{\sigma}}$ is a configuration. We now check that $\caus{\sigma}$ is an augmentation.
	
	\emph{Rule-abiding.} Let $t \in \ev{\caus{\sigma}}$ and $s\leq_{\deseq{\caus{\sigma}}}t$. By definition, there exists a chain of pointers from $t_\omega$ to $s_\omega$. This implies that $s\prefix t$, hence $s\leq_{\caus{\sigma}} t$.
	
	\emph{Courteous.} Let $s \imc_{\caus{\sigma}} t$ such that $\lambda_A(s)=+$ or $\lambda_A(t)=-$. By definition of $\leq_{{\caus{\sigma}}}$, $s=s'b^+$ and $t=s'ba^-$ since $A$ is alternating. By definition of a P-view, $a$ points to $b$ and $s\imc_{\deseq{\caus{\sigma}}} t$.
	
	\emph{Deterministic.} Let $s^-\imc_{\caus{\sigma}} t_1^+$ and $s^-\imc_{\caus{\sigma}} t_2^+$. By definition of $\imc_{\caus{\sigma}}$, we have $s=s'a^-$, $t_1 = s'ab_1^+$, $t_2=s'ab_2^+$, and $b_1=b_2$ (with the same pointer) by determinism of $\sigma$.
	
	So $\caus{\sigma}$ is an augmentation. Finally, we check the two additional conditions.
	
	\emph{Receptive.} Let $s\in\ev{\caus{\sigma}}$ such that $\partial_{\caus{\sigma}}(s)\imc_A a^-$. Since $A$ is alternating, we know $s=s'b^+$, and $s\in\sigma$ by definition of $\caus{\sigma}$. Furthermore, $\partial_{\caus{\sigma}}(s)\imc_A a^-$ implies that $sa\in\Plays(A)$, where $a$ points to $s_\omega$. So $sa\in\ev{\caus{\sigma}}$, with $s\imc_{\caus{\sigma}} sa$ and $\partial_{\caus{\sigma}}(sa)=a$.
	
	\emph{$-$-linear.} Let $s\imc_{{\caus{\sigma}}} t_1^-$ and $s\imc_{{\caus{\sigma}}} t_2^-$ such that $\partial_{\caus{\sigma}}(t_1)=\partial_{\caus{\sigma}}(t_2)$. By definition of $\leq_{\caus{\sigma}}$, $t_1=sa_1^-$ and $t_2=sa_2^-$, where both $a_i$ point to $s_\omega$ by courtesy. Moreover, by definition of $\partial_{\caus{\sigma}}$, $a_1 = a_2$. Hence $t_1=t_2$.
\end{proof}

\begin{lemma}\label{lem:sig_tot_caus_tot}
	For $\sigma : A$ finite innocent on $A$ well-opened, ${\caus{\sigma}}$ is total iff $\sigma$ is total.
\end{lemma}
\begin{proof}
 From the previous lemma, $\caus{\sigma}$ is always receptive.
 
 \emph{If.} Assume $\sigma$ is total. Consider $s\in\ev{\caus{\sigma}}$ maximal for $\leq_{\caus{\sigma}}$, and assume $\lambda(s)=-$. Then by definition $s$ is a P-view and $s=s'a^-$, where $s'\in\sigma$ by innocence and $s'a\in\Plays(A)$. By totality of $\sigma$, there exists $b^+$ such that $s'ab\in\sigma$, and since $s$ is already a P-view we have $\pview{s'ab}=s'ab$. Hence $s\imc_{{\caus{\sigma}}}s'ab$, which contradicts the maximality of $s$.
 
 \emph{Only if.} Assume $\caus{\sigma}$ is total. Consider $s\in\sigma$ such that $sa\in\Plays(A)$. Let $t=\pview{s}$, then $t\in\sigma$ by innocence of $\sigma$. Moreover, $ta=\pview{sa}\in\ev{\caus{\sigma}}$. By totality of $\caus{\sigma}$, $ta$ cannot be maximal, so there exists $b^+$ such that $ta\imc_{{\caus{\sigma}}}tab$. Since $tab\in\ev{\caus{\sigma}}$ and $tab=\pview{tab}$, we also have $tab\in\sigma$. By innocence of $\sigma$, $sab\in\sigma$, and $\sigma$ is total.
\end{proof}

Those two lemmas prove Proposition \ref{prop:pview_aug}.

In the development, we focus on expansions of causal strategies (Definition \ref{def:exp}). We prove minimality-preservation and uniqueness of morphisms.

\begin{lemma}\label{lem:aug_mu_pres_min}
	Consider $\q \in \Aug(A)$ and $a\in\ev{\q}$.
	
	Then $\partial_\q(a) \in \min(A)$ iff $a = \init(\q)$.
\end{lemma}
\begin{proof}
	\emph{If.} Assume that $a = \init(\q)$. Assuming that $\partial_\q(a)$ is not minimal, then by condition \emph{minimality-preserving} of Definition \ref{def:conf}, $a$ is not minimal in $\deseq\q$. But then by condition \emph{rule-abiding} of Definition \ref{def:augmentation},
	$a$ cannot be minimal in $\q$ either, contradiction.
	
	\emph{Only if.} Assume $\partial_\q(a) \in \min(A)$. In particular, $a$ has
	negative polarity since $A$ is negative. Assuming that $a$ is not minimal in $\q$, then it has an antecedent $a' \imc_\q a$. By courtesy, we have $a' \imc_{\deseq \q} a$ as well. Hence, $\partial_\q(a') \imc_A \partial_\q(a)$, contradiction.
\end{proof}
In particular, this proves that morphisms preserve initial events (by arena-preservation).

%

\begin{lemma} \label{lem:mor_unique}
	Consider $\p\in\Aug(A)$ a causal strategy and $\q\in\exp(\p)$. 
	
	Then there exists a unique morphism $\varphi:\q\to\p$.
\end{lemma}
\begin{proof}
	The existence is given by the definition of $\q\in\exp(\p)$. 
	
	Consider $\varphi,\psi:\q\to\p$ two morphisms from $\q$ to $\p$.
Consider $a\in\ev{\q}$ minimal such that $\varphi(a)\neq\psi(a)$. If $a$ is minimal for $\leq_\q$, this contradicts Lemma \ref{lem:aug_mu_pres_min} (by arena-preservation). Therefore, there is a (necessarily unique) antecedent $a'\imc_\q a$ for which $\varphi(a') = \psi(a')$. By causality-preservation of morphisms, we have $\varphi(a') \imc_\p \varphi(a)$ and $\varphi(a') \imc_\p \psi(a)$. 
	If $a$ is positive, then $\varphi(a) = \psi(a)$ by determinism of $\p$, contradiction. 	
	If $a$ is negative, then likewise we remark that $\partial_\p(\varphi(a))=\partial_\q(a)$ and $\partial_\p(\psi(a))=\partial_\q(a)$, contradicting $-$-linearity of $\p$.
\end{proof}
This lemma ensures the uniqueness condition of Definition \ref{def:exp}.


\begin{lemma}\label{lem:conf_to_exp}
	Let $\sigma:A$ a total finite innocent strategy and $\x\in\positions{\sigma}$. Then $\x\in\positions{\caus{\sigma}}$.
\end{lemma}
\begin{proof}
	Consider $\x\in\positions{\sigma}$, it is the isomorphism class of $\deseq{s}$ for some $s=s_1\ldots s_n\in\sigma$. We build an augmentation $\q(s)$ as follows: its underlying configuration is $\deseq{\q(s)}=\deseq{s}$, with $\ev{\q(s)}=\{1,\ldots,n\}$. Its causal order is defined inductively with for any $i,j\in\ev{\q(s)}$,
	\[
	\begin{array}{lllllll}
	i\leq_{\q(s)} 2j 
	&\Leftrightarrow
	&i=2j 
	&&\text{or}
	&&i\leq_{\q(s)} 2j-1 \;;\\
	i\leq_{\q(s)} 2j+1
	& \Leftrightarrow
	& i=2j+1 
	&& \text{or}
	&& i\leq_{\q(s)} k
	\text{ where $s_{2j+1}$ points to $s_k$ in $s$.}
	\end{array}
	\]
	Remark that since $s$ is alternating and negative, for any $i\in\ev{\q(s)}$, $s_i$ is negative iff $i$ is odd. So $i\leq_{\q(s)}j$ means that $s_i$ is reached in the computation of $\pview{s_{\leq j}}$.
	It is clear that $\tuple{\ev{\q(s)}, \leq_{\q(s)}}$ is a tree; we must check that $\q(s)$ is an augmentation.
	
	\emph{Rule-abiding.} For any $i\leq_{\deseq{\q(s)}}j$, there exists a chain of justifiers from $s_i$ to $s_j$. Since $s$ is P-visible, $s_i$ must appear in $\pview{s_1\ldots s_j}$. But $\leq_{\q(s)}$ is inductively defined to follow the construction of a P-view, so $i\leq_{\q(s)} j$.
	
	\emph{Courteous.} For all $i\imc_{\q(s)} j$, if $\pol(s_i)=+$ (resp. $\pol(s_j)=-$), then $i$ is even (resp. $j$ is odd), and by definition of $\leq_{\q(s)}$, we know that $s_j$ points to $s_i$. Hence $i\imc_{\deseq{\q(s)}}j$.
	
	\emph{Deterministic.} For any negative $i$, then $i$ is odd and $i\imc_{\q(s)} j$ implies $j=i+1$.
	
	
	We must now check that $\q(s)\in\exp(\caus{\sigma})$. Consider $\varphi:\ev{\q(s)}\to\ev{\caus{\sigma}}$ such that $\varphi(i)=\pview{s_{\leq i}}$. Then $\varphi$ is a morphism:
	
	\emph{Arena-preserving.} It is clear that $\partial_{\caus{\sigma}}\circ \varphi = \partial_{\q(s)}$, since the P-view preserves the last move.
	
	\emph{Configuration-preserving.} For all $i\imc_{\deseq{\q(s)}}j$, we know that $s_j$ points to $s_i$ since $\deseq{q(s)}=\deseq{s}$. By innocence of $\sigma$, $s$ is P-visible and $s_j$ keeps its pointer in $\pview{s_{\leq j}}$ (and this pointer is still $s_i$). Hence, $\varphi(i)\imc_{\deseq{\caus{\sigma}}}\varphi(j)$. 
	
	\emph{Causality-preserving.} By definition, for any $j\in\ev{\q(s)}$ such that $j\geq 1$, we have $2j-1\imc_{\q(s)}2j$. Moreover, $s_{2j}$ is positive, so $\varphi(2j) = \pview{s_{\leq 2j}} =  \pview{s_{\leq 2j-1}}s_{2j}$, where $s_2j$ keeps its pointer. So $\varphi(2j-1)\imc_{{\caus{\sigma}}}\varphi(2j)$.
	Likewise, for any $i,2j+1\in\ev{\q(s)}$, we have $i\imc_\q 2j+1$ iff $s_{2j+1}$ points to $s_i$ in $s$. In that case, we have $\varphi(2j+1)=\pview{s_{\leq 2j+1}} = \pview{s_{\leq i}} s_{2j+1}^-$, where $s_{2j+1}$ is negative since $2j+1$ is odd. So $\varphi(i)\imc_{\caus{\sigma}}\varphi(2j+1)$.
	
	\emph{$+$-obsessional.} For all $i^-\in\ev{\q(s)}$, if $\varphi(i)\imc_{\caus{\sigma}}t^+$, then it means $t = \pview{s_{\leq i}} b^+ \in\sigma$. Since $s\in\sigma$, $s$ is of even length, so $i+1\in\ev{\q}$ and by determinism and innocence, $s_{i+1}=b$. Therefore, $\varphi(i+1)=t$, and $i\imc_{\q(s)}i+1$.
	
	Finally, $\q(s)\in\exp(\caus{\sigma})$, which means
$\x=\pos{\q(s)}\in\positions{\caus{\sigma}}$.
\end{proof}

\begin{lemma}\label{lem:exp_to_seq}
	Consider $\sigma:A$ a total finite innocent strategy, and $\x\in\positions{\caus{\sigma}}$. Then $\x\in\positions{\sigma}$.
\end{lemma}
\begin{proof}
	Consider $\x\in\positions{\caus{\sigma}}$, then there exists $\q\in\exp(\caus{\sigma})$ such that $\pos{\q}=\x$. From the totality of $\caus{\sigma}$ (Lemma \ref{lem:sig_tot_caus_tot}), $\q$ has maximal events all positive - hence by determinism it has exactly as many Player as Opponent moves. Moreover, by courtesy, any $a\imc_\q a'$ implies that $\lambda(a)\neq\lambda(a')$, since $A$ is alternating. Therefore, there exists an alternating sequentialization of $\q$, which we construct inductively. We start by the minimal event of $\q$, negative by Lemma \ref{lem:aug_mu_pres_min}, and its only successor (exists by totality, unique by determinism). If a positive event $b^+$ has one or more successors, we inductively construct sequentializations for each subtree of root $a^-\in\suc(b)$. All sequences are alternating, start with a negative event, and end with a positive event. Hence, we can concatenate them in an arbitrary order. We obtain an alternating sequence $s_1^-, s_2^+, \ldots, s_n^+$ such that:
	\[
	\ev{\q}=\{s_i \tq 1\leq i \leq n\}.
	\]
	We can see $s:= s_1\ldots s_n$, with pointers imported from $\deseq{\q}$ (\emph{i.e.} $\deseq{s}=\deseq{\q}$), as a play on arena $\deseq{\q}$. Then for any $1\leq i \leq n$, it holds that $\pview{s_{\leq i}}$ coincides with $[s_i]_\q$, the causal dependency of $s_i$ in $\q$ (clear by induction on $i$).
	Writing $\partial_\q(s) = \partial_\q(s_1)\ldots\partial_\q(s_n)$ with pointers imported from $\deseq{\q}$, it follows that $\partial_\q(s)\in\Plays(A)$ and $\pview{(\partial_\q(s))_{\leq i}}\in\pviews{\sigma}$ (by definition of $\caus{\sigma}$). Hence $\partial_\q(s)\in\sigma$ by innocence, with $\deseq{\partial_\q(s)}\iso\deseq{s}$. Therefore, $\deseq{\partial_\q(s)}\iso\deseq{\q}$, \emph{i.e.}  $\pos{\q}=\pos{\partial_\q(s)}\in\positions{\sigma}$.	
\end{proof}

\possigposcaus*
\begin{proof}
	Immediate by Lemmas \ref{lem:conf_to_exp} and \ref{lem:exp_to_seq}.
\end{proof}

\section{Positional Injectivity: Proofs from Section \ref{sec:pos_inj}}

\subsection{Compositional Properties of Bisimulations (Section \ref{subsec:comp_bissim})}\label{app:comp_bissim}

In this section, we prove some technical lemmas about bisimulations.

Recall that for $\q\in\Aug(A)$, $a\in\ev{\q}$ there is an
order-iso $\Gamma\tuple{a} : [a]_\q^- \iso [\varphi(a)]_\p^-$, with $[a]^-_\q$
the totally ordered set of negative dependencies of $a$ in $\q$. Recall
also that the \emph{co-depth} of $a\in \ev{\q}$ is the maximal length
of a causal chain $a = a_1 \imc_\q \dots \imc_\q a_k$ in $\q$.

\begin{lemma}\label{lem:oexp_to_sim}
	Consider $\q, \p \in \Aug(A)$ where $\p$ is a causal strategy and $\q
	\in \exp(\p)$ is a $-$-obsessional expansion with the morphism
	$\varphi : \q\to\p$.
	Then 
	\[
	a \sim_{\Gamma\tuple{a}} \varphi(a)\,.
	\]
\end{lemma}
\begin{proof}
	By induction on the co-depth of $a\in\ev{\q}$. 
	We must check that $a \sim_{\Gamma\tuple{a}} \varphi(a)$.
	
	First, \emph{(a)} is immediate by the definition of $\Gamma\tuple{a}$.
	
	\emph{(b).} We know that if $a$ is positive, then $\just(a^+)\in[a]_\q^-=\dom(\Gamma\tuple{a})$ and $\just(\varphi(a))\in[\varphi(a)]^-_\p=\cod(\Gamma\tuple{a})$. Moreover
	$\just(\varphi(a))=\varphi(\just(a))$ since $\varphi$ is \emph{configuration-preserving}.

	\emph{(1).} Assume $a^+ \imc_\q b^-$. Then $\varphi(a)\imc_\p \varphi(b)$. By induction hypothesis, $b\sim_{\Gamma(b)}\varphi(b)$. But $[b^-]^-_\q=[a]^-_\q\cup\{b\}$ and $[\varphi(b)^-]^-_\p=[\varphi(a)]^-_\p\cup\{\varphi(b)\}$, so finally
	\[
	b \sim_{\Gamma\tuple{a}\cup\{(b,\varphi(b))\}} \varphi(b)\;.
	\]
	The same reasoning applies for the symmetric condition. Assume $\varphi(a)^+\imc_\p b$, then  $\varphi^{-1}(b)$ exists by receptivity of $\q$ and $-$-linearity of $\p$.
	
	\emph{(2).} Same as for \emph{(1)}, except $[b^+]^-_\q=[a]^-_\q$ and $[\varphi(b)^+]^-_\p=[\varphi(a)]^-_\p$. The same reasoning applies for the symmetric condition. Assume $\varphi(a)^-\imc_\p b$, then  $\varphi^{-1}(b)$ exists by $+$-obsessionality of $\q$.
\end{proof}

In the following proofs, we will need a few additional properties on bisimulations and contexts, to define a \emph{minimal context}. 

\begin{lemma}\label{lem:matching_inded_context}
	Consider $\q,\p\in\Aug(A)$ with $\varphi:\q\iso\p$. Consider
	$a\sim_\Gamma^\varphi b$ for some $\Gamma$. 
	
	Then for any $a'\in\up a$, there exists $b'\in\up b$ such that $a'\sim_{\Gamma\cup\Delta}^\varphi b'$, where
	\[
	a=a_0 \imc_\q a_1 \imc_\q \ldots \imc_\q a'=a_n \,,
	\qquad
	\qquad	
	b=b_0 \imc_\p b_1 \imc_\p \ldots \imc_\p b'=b_n \,,
	\]
	and $\Delta$ is the context defined as 
	$\Delta = \{(a_i,b_i) \tq 0\leq i \leq n \text{ and } \pol(a_i)=-\}$.
		
	Moreover, if $a\sim_{\Gamma'}^\varphi b$ for a context $\Gamma'$, we also have $a'\sim_{\Gamma'\cup\Delta}^\varphi b'$.
\end{lemma}
\begin{proof}
Immediate by induction of the co-depth of $a'$ and by definition of bisimulation.
\end{proof}

\begin{definition}\label{def:minimal}
	Consider $\q, \p \in \Aug(A)$, $\varphi : \deseq \q \iso \deseq \p$, 
	$a \in \ev{\q}$, $b \in \ev{\p}$ with $a \sim_\Gamma^\varphi b$ for some context $\Gamma$.		
	We define $\Gamma_{a,b}$ the \textbf{minimal context} for
	$a\sim^\varphi_\Gamma b$ as the restriction of $\Gamma$ s.t.
	\[
	\begin{array}{c}
		c\in\dom(\Gamma_{a,b}) \quad \Leftrightarrow \quad
		\begin{cases}
			\exists a'\in\up a, \just(a') = c & \text{\emph{(a)}} \\
			\Gamma(c) \neq \varphi(c) & \text{\emph{(b)}}
		\end{cases}
	\end{array}
	\]
	for all $c\in\ev{\q}$,and symmetrically the mirror condition applies to any $d\in\ev{\p}$.
\end{definition}

%
\begin{lemma}\label{lem:minimal}
	Consider $\q, \p \in \Aug(A)$ with $\varphi : \deseq \q \iso \deseq \p$. Consider
	$a \in \ev{\q}$, $b \in \ev{\p}$ and $\Gamma$, $\Gamma'$ two contexts such that $a \sim_\Gamma^\varphi b$ and $a\sim_{\Gamma'}^\varphi b$.
	
	Then $\Gamma_{a, b} = \Gamma'_{a, b}$. Moreover, $\Gamma_{a,b}$
	is the minimal (for inclusion) context s.t.
	$a\sim^\varphi_{\Gamma_{a,b}}b$. 
\end{lemma}
\begin{proof}
	The equality comes from Lemma \ref{lem:matching_inded_context} and the
	definition of $\Gamma_{a,b}$ and $\Gamma'_{a,b}$. By induction,
	$a\sim_{\Gamma_{a,b}}^\varphi b$, since we can safely remove from
	$\dom(\Gamma)$ all $c$ that are never ``used'', \emph{i.e.} such that
	there exists no $a'\in\up a$ having $c$ as pointer; and all $c$ such
	that $\Gamma(c)=\varphi(c)$, because then we can use condition $(c)$ of
	Definition \ref{def:main_bisim} instead of condition $(b)$.  
	Finally, for any context $\Gamma''$ such that $a\sim_{\Gamma''}^\varphi b$, we have $\Gamma_{a,b}=\Gamma''_{a,b}\subseteq\Gamma''$, so $\Gamma_{a,b}$ is minimal for inclusion.
\end{proof}

This lemma allows us to write \emph{the minimal context for $a$, $b$} without mentioning $\Gamma$.

\subsection{Clones (Section \ref{subsec:clones})}\label{app:clones}

A key notion in the proof of positional injectivity is the notion of clones, a variation of bisimulation. Although the added constraint on contexts makes transitivity more challenging, we can still prove a variation of Lemma \ref{lem:bis_eq}. We use the same notation that for the usual bisimulation: for any $a, b$ events of an augmentation $\q$, $a\clone b$ means $a\clone^\id b$.

\begin{lemma}\label{lem:clone_trans}
	Consider $\q, \p, \r \in \Aug(A)$, $\varphi : \deseq \q \iso \deseq \p$
	and $\psi : \deseq \p \iso \deseq \r$, and:
	\[
	a \clone^\varphi b
	\qquad
	\qquad
	b \clone^\psi c
	\]
	for some $a \in \ev{\q}$, $b \in \ev{\p}$, and $c \in \ev{\r}$.
	Then, we also have $a \clone^{\psi \circ \varphi} c$.
\end{lemma}
\begin{proof}
	Consider $\Gamma_1$ and $\Gamma_2$ the minimal contexts such that $a\sim^\varphi_{\Gamma_1} b$ and $b\sim_{\Gamma_2}^\psi c$. 
	If $\cod(\Gamma_1) = \dom(\Gamma_2)$, the result is immediate by Lemma \ref{lem:bis_eq}: we get $a\sim_{\Gamma_2\circ\Gamma_1}^{\psi\circ\varphi}c$ with, for any $e\in\dom(\Gamma_1)=\dom(\Gamma_2 \circ \Gamma_1)$,
	\[
	\psi(\varphi(\just(e))) = \psi(\just(\Gamma_1(e))) = \just(\Gamma_2(\Gamma_1(e)))
	\]
	so $\Gamma_2\circ\Gamma_1$ preserves pointers and $a \clone^{\psi \circ \varphi} c$.
	
	Now, assume there exists $e\in\cod(\Gamma_1)$ such that $e\notin \dom(\Gamma_2)$. Since $\Gamma_1$ is minimal, there exists $b'\in\up b$ such that $\just(b')=e$. By $b\sim^\psi_{\Gamma_2} c$ and Lemma \ref{lem:matching_inded_context}, there exists a matching $c'\in\up c$ such that $b'\sim^\psi_{\Gamma_2\cup\Delta}$, with $\Delta$ the negative moves between $b$ and $b'$, paired with the negative moves between $c$ and $c'$. 
	Since $e\in\cod(\Gamma_1)$, $\neg (e\geq_\p b)$, so $e\notin \dom(\Delta)$. 
	Hence, 
	$\just(c')=\psi(e)$ and $\psi(e)\notin\cod(\Gamma_2)$. 
	So we can write
	\[
	b \sim_{\Gamma_2 \cup \{(e,\psi(e))\}}^\psi c
	\]
	where $\Gamma_2 \cup \{e,\psi(e)\}$ preserves pointers.
	Likewise, for any $e'\in\dom(\Gamma_2)$, $e'\notin\cod(\Gamma_1)$, we have $\Gamma_1 \cup \{\varphi^{-1}(e'),e'\}$ well-defined and pointer-preserving, such that
	\[
		a \sim^\varphi_{\Gamma_1 \cup \{(\varphi^{-1}(e'),e')\}} b\;.
	\]
	
	This allows us to define the following pointer-preserving contexts:
	\[
	\begin{array}{rl}
	\Gamma'_1 & := \Gamma_1 \cup 
		\{(\varphi^{-1}(e'),e')\tq e'\in\dom(\Gamma_2), e'\notin\cod(\Gamma_1)\}\\
	\Gamma'_2 & := \Gamma_2 \cup 
		\{(e,\psi(e)) \tq e\in\cod(\Gamma_1), e\notin \dom(\Gamma_2)\}
	\end{array}
	\]
	
	Then $\Gamma'_2\circ\Gamma'_1$ preserves pointers, and by Lemma \ref{lem:bis_eq} we have $a\sim_{\Gamma'_2\circ\Gamma'_1}^{\psi\circ\varphi}c$, so $a\clone^{\psi\circ\varphi} c$.	
\end{proof}

This allows us to prove equivalence properties for the clone relation.
\begin{lemma}\label{lem:clone_eq}
	Consider $\q, \p, \r \in \Aug(A)$ augmentations, with $\varphi :
	\deseq{\q} \iso \deseq{\p}$ and $\psi : \deseq{\p} \iso \deseq{\r}$ two  isomorphisms, and
	events $a \in \ev{\q}$, $b \in \ev{\p}$, $c \in \ev{\r}$:
	\[
	\begin{array}{rl}
	\text{\emph{reflexivity:}} &
	\text{$a \clone^\id a$,}\\ 
	\text{\emph{transitivity:}} &
	\text{if $a \clone^\varphi b$ and $b \clone^\psi c$, then $a
		\clone^{\psi \circ \varphi} c$,}\\ 
	\text{\emph{symmetry:}} &
	\text{if $a \clone^\varphi b$ then $b
		\clone^{\varphi^{-1}} a$.}
	\end{array}
	\]
\end{lemma}
\begin{proof}
	\emph{Reflexivity.} By Lemma \ref{lem:bis_eq}, $a\sim^\id a$, which implies $a\clone^\id a$.
	
	\emph{Transitivity.} See Lemma \ref{lem:clone_trans}.
	
	\emph{Symmetry.} Immediate by Lemma \ref{lem:bis_eq}: if $\Gamma$ preserves pointers, so does $\Gamma^{-1}$.
\end{proof}

Clones through $\id$ in characteristic expansions will be especially
interesting, because then we can partition equivalence classes of
$\clone^\id$ into successors of forks.

\twinclones*
%
\begin{proof}
	If $X = \{\init(\q)\}$, $a_1=a_2$ and the result is immediate by determinism and reflexivity. Otherwise, we note $b$ the predecessor of both $a_1$ and $a_2$ for $\leq_\q$.
	
	First, we prove that $b_1$ and $b_2$ are bisimilar. Since $\q$ is a $-$-obsessional expansion of $\p$, there exists a (unique, by Lemma \ref{lem:mor_unique}) morphism $\varphi : \q \to \p$. By Lemma \ref{lem:oexp_to_sim},
	\[
	b_1 \sim_{\Gamma(b_1)} \varphi(b_1)
	\qquad \text{and}\qquad
	b_2 \sim_{\Gamma(b_2)} \varphi(b_2)\;.
	\]
	
	By $-$-linearity of $\p$, we know that $\varphi(a_1)=\varphi(a_2)$, which implies $\varphi(b_1)=\varphi(b_2)$ by determinism. So $\cod(\Gamma(b_1))=\cod(\Gamma(b_2))$, and by Lemma \ref{lem:bis_eq},
	\[
	b_1\sim_{\Gamma(b_2)^{-1}\circ\Gamma(b_1)}b_2\;.
	\]
	
	Writing $\Gamma_1:=\Gamma(b_1)$, $\Gamma_2:=\Gamma(b_2)$ and
$\Gamma:=\Gamma_2^{-1} \circ \Gamma_1$, it remains to check
that $\Gamma$ preserves pointers. Consider $c\in \dom(\Gamma) =
[b_1]^-_\q$, either $c = a_1$ or $c\leq_\q b$. If $c = a_1$, then
$\Gamma_1(a_1)=\varphi(a_1)=\Gamma_2(a_2)$,
so $\Gamma(a_1)=a_2$ (and both have the same pointer $b$ by courtesy).
Otherwise, $c\leq_\q b$, so $c\in\dom(\Gamma_2)$ and
$\Gamma_1(c)=\varphi(c)=\Gamma_2(c)$, hence $\Gamma(c)=c$.
In both cases, $\Gamma$ preserves pointers, so finally $b_1 \clone b_2$.
\end{proof}

\subsection{Positional Injectivity (Section \ref{subsec:pos_inj})}
In this section, we prove additional lemmas needed in the proof of Lemma \ref{lem:key}.

\begin{lemma}\label{lem:causal_depth}
	Consider $\q,\p\in\Aug(A)$ two augmentations such that there
exists an isomorphism $\varphi:\deseq\q\iso\deseq\p$. Consider
$a\in\ev{\q}$, $b\in\ev{\p}$ and $\Gamma$ a pointing context s.t. $a\sim_\Gamma^\varphi b$. 

Then $a$ and $b$ have the same co-depth.	
\end{lemma}
\begin{proof}
	Straightforward by induction.
\end{proof}

\begin{lemma}\label{lem:context_add}
	Consider $\q,\p\in\Aug(A)$ two augmentations such that there exists an isomorphism $\varphi:\deseq\q\iso\deseq\p$. Consider $a\in\ev{\q}$, $b\in\ev{\p}$ and $\Gamma$ a pointing context such that $a\sim_\Gamma^\varphi b$. Then for any $c\in\ev{\q}$ such that $c\notin\dom(\Gamma)$, $\varphi(c)\notin\cod(\Gamma)$ and $\neg (c\in\up a)$, $\neg (\varphi(c)\in\up b)$, we have $a\sim_{\Gamma \cup \{(c,\varphi(c))\}}^\varphi b$. Moreover, for any $c\in\ev{\q}$ and $d\in\ev{\p}$ such that 
	\[
	\begin{array}{ccccc}
		c\notin\dom(\Gamma)\;, & \quad & \neg (c\in\up a)\;,
		& \quad & \forall a'\in\up a, \just(a')\neq c\;,\\
		d\notin\cod(\Gamma)\;, & \quad & \neg (d\in\up b)\;,
		& \quad & \forall b'\in\up b, \just(b')\neq d\;,\\
	\end{array}
	\]
	then we also have $a\sim_{\Gamma \cup \{(c,d)\}}^\varphi b$.
\end{lemma}
\begin{proof}
	Straightforward by induction. Either $c$ is never used in the bisimulation (no one in $\up a$ points to $c$), and we can pair it with any $d$ which is not used either and add $(c,d)$ to $\Gamma$ (as long as we still have $(\Gamma\cup\{(c,d)\})\vdash (a,b)$); or it is used with condition \emph{(c)} of Definition \ref{def:main_bisim} and we can add $(c,\varphi(c))$ to $\Gamma$ and use condition \emph{(b)} instead.
\end{proof}

\begin{lemma}\label{lem:minimal_context_clone}
	Consider $\q\in \Aug(A)$ and $a, b \in \ev{\q}$ such that $a \clone b$.

	Then the minimal context for $a$ and $b$ is either empty or $\Gamma:\{c\}\iso\{b\}$.
\end{lemma}
\begin{proof}
	Assume, seeking a contradiction, that the minimal context $\Gamma$ has at least two distinct elements $c_1, c_2 \in\dom(\Gamma)$. First, we can remark that since $a\clone b$, there exists $\Gamma'$ a pointers-preserving context such that $a\sim_{\Gamma'} b$, and since $\Gamma$ is a restriction of $\Gamma'$, $\Gamma$ also preserves pointers.	
	
	By condition \emph{(a)} of Definition \ref{def:minimal}, $c_1 \leq_\q a$ and $c_2 \leq_\q a$. Therefore, $c_1\leq_\q c_2$ or $c_2 \leq_\q c_1$ --
	assume \emph{w.l.o.g.} that it is the former. By courtesy, $\just(c_1) \leq_\q \just(c_2)$ as well. For the same reason, $\Gamma(c_1) \leq_\q \Gamma(c_2)$ or $\Gamma(c_2) \leq_\q \Gamma(c_1)$. 
	
	If it is the latter, this entails that $\just(\Gamma(c_2)) \leq_\q \just(\Gamma(c_1))$ by courtesy; \emph{i.e.}, since $\Gamma$ preserves pointers, $\just(c_2)\leq_\q \just(c_1)$. So $\just(c_1) = \just(c_2)$, and because $c_1, c_2
	\leq_\q a$, we have $c_1 = c_2$, contradiction. 
	
	So, $\Gamma(c_1) \leq_\q \Gamma(c_2)$, and $\Gamma(c_1)\neq\Gamma(c_2)$ by hypothesis. By courtesy, this entails that $\Gamma(c_1) \leq_\q \just(\Gamma(c_2))$. Likewise, $c_1 \leq_\q c_2$ entails $c_1 \leq_\q \just(c_2)$. Moreover, $\Gamma$ preserves pointers, so $\just(c_2) = \just(\Gamma(c_2))$. Hence, we have both
	\[
	\Gamma(c_1) \leq_\q \just(c_2)
	\qquad
	\qquad
	c_1 \leq_\q \just(c_2)\,,
	\]
	so $c_1$ and $\Gamma(c_1)$ are comparable for $\leq_\q$ since $\q$ is a forest. But they are negative, so they have the same antecedent by courtesy. This implies $c_1 = \Gamma(c_1)$, which contradicts condition \emph{(b)} of Definition \ref{def:minimal}.
\end{proof}


\begin{lemma}\label{lem:clone_lift}
	Consider $\q,\p\in\Aug(A)$, $\varphi:\deseq\q\iso\deseq\p$.
	Consider also $a^+\in\ev{\q}$ s.t. $\suc(a)=\bigcup_{i\in I} G_i$,
	where $I\subseteq\N$ and for $i\in I$, $G_i = \{b_{i,1},\ldots,
	b_{i,2^i} \}\in\twin(\q)$ with $\sharp G_i = 2^i$.
	
	Then we have $a\clone^\varphi\varphi(a)$, provided the two conditions
	hold:
	\begin{align}
		\text{if $b_{i,j}\imc_\q c_{i,j}$, then $\varphi(b_{i,j})\imc_\p d_{i,j}$ and $c_{i,j}\clone^\varphi d_{i,j}$}\;, \label{eq:lift} \\
		\text{if $\varphi(b_{i,j})\imc_\p d_{i,j}$, then
			$b_{i,j}\imc_\q c_{i,j}$ and $c_{i,j}\clone^\varphi d_{i,j}$}\;. \label{eq:lift_sym}
	\end{align}
\end{lemma}

\begin{proof}
	First, remark that $\partial_{\q}(a)=\partial_\p(\varphi(a))$ and $\varphi(\just(a))=\just(\varphi(a))$.
	
	For any $i\in I$, $1\leq j \leq 2^i$, let $\Gamma_{i,j}$ be the
	minimal context for $b_{i,j}$ and $\varphi(b_{i,j})$. Such a context
	exists since either $b_{i,j}$ has no successors, and by
	\eqref{eq:lift_sym} neither does $\varphi(b_{i,j})$, either $b_{i,j}$ has
	only one (by determinism) and $c_{i,j}\clone^\varphi d_{i,j}$ by
	\eqref{eq:lift}. In both cases, $b_{i,j}\clone^\varphi \varphi(b_{i,j})$.
	
	We wish to take the union of all $\Gamma_{i,j}$ as the context for $a$
	and $\varphi(a)$, but this is only possible if they are
	\emph{compatible}. More precisely, we must ensure
	that for all $e\in\q, i,k\in I, 1\leq j \leq 2^i$
	and $1\leq l \leq 2^k$, if there are $c'_{i,j}\in\up b_{i,j}$ and
	$c'_{k,l}\in\up b_{k,l}$ having both $e$ as justifier, then their
	matching $d'_{i,j}\in\up\varphi(b_{i,j})$ and
	$d'_{k,l}\in\up\varphi(b_{k,l})$ also have the same justifier. This can
	only be a problem if $e$ appears in $\dom(\Gamma_{i,j})$ or in
	$\dom(\Gamma_{k,l})$ as otherwise both justifiers are $\varphi(e)$.
	
	For all $i,j$, $\Gamma_{i,j}$ has either one or zero element by Lemma
	\ref{lem:minimal_context_clone}. If all $\Gamma_{i,j}$ are empty, we can
	directly lift the clone relation to $a$.
	Otherwise, consider $i,j$ s.t. $\Gamma_{i,j}:\{e_{i,j}\}\iso\{f_{i,j}\}$.	
	From Definition \ref{def:minimal}, $e_{i,j} \in [b_{i,j}]^-_\q$ and
	$f_{i,j} \in [\varphi(b_{i,j})]^-_\p$. Actually we have
	$f_{i,j}\in [\varphi(a)]^-_\p$: indeed $f_{i,j}\neq
	\varphi(b_{i,j})$, since $e_{i,j}$ and $f_{i,j}$ have the same
	justifier through $\varphi$ and the only $e\in[b_{i,j}]^-_{\q}$ s.t.
	$\varphi(\just(e))= \just(\varphi(b_{i,j}))$ is $b_{i,j}$, which
	contradicts Definition \ref{def:minimal}. 
	
	Now, assume that for some $k,l$, there exists $c'_{k,l}\in\up b_{k,l}$
	s.t. $\just(c'_{k,l})=e_{i,j}$. Since $b_{k,l}\clone^\varphi
	\varphi(b_{k,l})$, there is a matching $d'_{k,l}\in\up
	\varphi(b_{k,l})$ such that 
	\[
	\varphi(\just(e_{i,j}))=\just(\just(d'_{k,l}))\,.
	\]

	For $b_{i,j}\sim^\varphi_{\Gamma_{i,j}} \varphi(b_{i,j})$ and
	$b_{k,l}\sim^\varphi_{\Gamma_{k,l}} \varphi(b_{k,l})$ to be compatible,
	we need $\just(d'_{k,l})=f_{i,j}$. But since $\Gamma_{i,j}$ preserves
	pointers, 
	\[
	\varphi(\just(e_{i,j}))=\just(f_{i,j})\,.
	\]

	Putting both equalities together, we obtain 
	\[
	\just(\just(d'_{k,l})) =
	\just(f_{i,j}) \,,
	\]
	where $\just(d'_{k,l}) \in [d'_{k,l}]^-_\p$ and
	$f_{i,j} \in [\varphi(a)]^-_\p$. But
	$[\varphi(a)]^-_\p\subseteq[d'_{k,l}]^-_\p$, which is a fully ordered
	set for $\leq_\p$, so $\just(d'_{k,l})$ and $f_{i,j}$ are comparable.
	Moreover, they are negative, so by courtesy
	\[
	\just(\just(d'_{k,l})) =
	\just(f_{i,j})$ iff $\pred(\just(d'_{k,l})) = \pred(f_{i,j}) \,
	\]
	where
	$\pred$ is the predecessor for $\leq_\p$. Hence,
	$\just(d'_{k,l})=f_{i,j}$ (see Figure \ref{fig:just_clones}, where
	$\imc$ represents $\imc_\q$, $\cdots$ represents $\imc_{\deseq{\q}}$,
	and $\rightarrow$ represents $\leq_\q$ (and the same applies for $\p$)).
	
	\begin{figure}
		\begin{minipage}{0.45\textwidth}
			\scalebox{1}{
				\[
				\xymatrix@R=2pt@C=5pt{
					& \just(e_{i,j})^+ \ar@{-|>}[d] & \\
					& e_{i,j} \ar@{.}@/_2.5pc/[dddddl] \ar@{.}@/^2.5pc/[dddddr] 
					\ar@{->}[dd] & \\
					& & \\
					& a^+ \ar@{-|>}[dl] \ar@{-|>}[dr] & \\
					b_{i,j} \ar@{->}[dd] & & b_{k,l}  \ar@{->}[dd] \\
					& & \\
					c'_{i,j} & & c'_{k,l}
				}
				\]
			}
		\end{minipage}
		\hfill
		\begin{minipage}{0.45\textwidth}
			\scalebox{1}{
				\[
				\xymatrix@R=2pt@C=5pt{
					& \just(f_{i,j})^+ \ar@{-|>}[dl] \ar@{-|>}[dr] & \\
					f_{i,j}\ar@{.}@/_1.7pc/[ddddd] \ar@{->}[ddr] &
					& \just(d'_{k,l}) \ar@{.}@/^1.7pc/[ddddd] \ar@{->}[ddl] \\
					&  & \\
					& \varphi(a)^+ \ar@{-|>}[dl] \ar@{-|>}[dr] & \\
					\varphi(b_{i,j}) \ar@{->}[dd] & & \varphi(b_{k,l}) \ar@{->}[dd] \\
					& & \\
					d'_{i,j} & & d'_{k,l}
				}
				\]
			}
		\end{minipage}
		\caption{Justifiers in $\q$ and $\p$}
		\label{fig:just_clones}
	\end{figure}
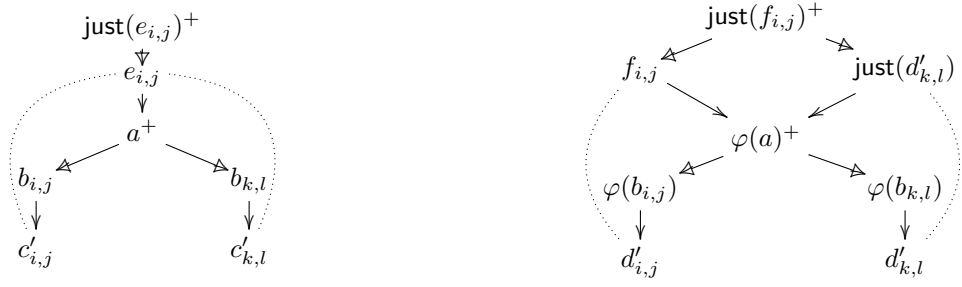
	
	So all contexts $\Gamma_{i,j}$ are compatible. Writing $\Gamma =
	\cup_{i,j}\Gamma_{i,j}$ it follows that $b_{i,j}\sim^\varphi_{\Gamma}
	\varphi(b_{i,j})$ via Lemma \ref{lem:context_add}; which entails that $a
	\sim^\varphi_{\Gamma} \varphi(a)$ by two steps of the bisimulation game.
	This implies $a\clone^\varphi \varphi(a)$ since all $\Gamma_{i,j}$
	preserve pointers. 
\end{proof}

\section{Beyond Total Finite Strategies: Proofs from Section \ref{sec:beyond_total_finite}}
\label{app:partial}

We now give the proof of Theorem \ref{th:partial}. 

Consider $\sigma_1, \sigma_2 : A$ finite (but not necessarily total)
innocent strategies. If they are empty, there is nothing to prove.
Otherwise, let $2n+2$ be the length of $s$ the longest P-view among them.
\emph{W.l.o.g.}, assume that $s\in\pviews{\sigma_1}$.
\begin{figure}[t]
\[
\raisebox{78pt}{$
\xymatrix@R=5pt@C=5pt{
&&&&\bq^-_0
	\ar@{-|>}[d]\\
&&&&\bq^+_0
	\ar@{.}[dll]
	\ar@{-|>}@/_.1pc/[dll]
	\ar@{.}[d]
	\ar@{-|>}@/_.1pc/[d]
	\ar@{.}[drr]
	\ar@{-|>}@/_.1pc/[drr]\\
\text{$n$ copies}
&&\bq^-_1
	\ar@{-|>}[d]&\dots
&\bq^-_1\ar@{-|>}[d]&\dots
&\bq^-_1\ar@{-|>}[d]\\
&&&&\bq^+_1
        \ar@{.}[dll]
        \ar@{-|>}@/_.1pc/[dll]
        \ar@{.}[d]
        \ar@{-|>}@/_.1pc/[d]
        \ar@{.}[drr]
        \ar@{-|>}@/_.1pc/[drr]&&~\\
\text{$n-1$ copies}
&&\bq^-_2
	\ar@{-|>}[d]&\dots
&\bq^-_2\ar@{-|>}[d]&\dots
&\bq^-_2\ar@{-|>}[d]\\
&&&&\dots	\ar@{-|>}[d]&&~\\
&&&&\bq^+_{n-1}
	\ar@{.}[d]
	\ar@{-|>}@/_.1pc/[d]\\
\text{$1$ copy}&&&&\bq^-_n
	\ar@{-|>}[d]\\
&&&&\bq^+_n
}$}
\qquad
\in 
\qquad
\exp\left(
\raisebox{78pt}{$
\xymatrix@R=5pt{
\bq^-_0	\ar@{-|>}[d]\\
\bq^+_0	\ar@{.}[d]
	\ar@{-|>}@/_.1pc/[d]\\
\bq^-_1	\ar@{-|>}[d]\\
\bq^+_1	\ar@{.}[d]
	\ar@{-|>}@/_.1pc/[d]\\
\bq^-_2	\ar@{-|>}[d]\\
\dots	\ar@{-|>}[d]\\
\bq^+_{n-1}
	\ar@{.}[d]
	\ar@{-|>}@/_.1pc/[d]\\
\bq^-_n	\ar@{-|>}[d]\\
\bq^+_n
}$}
\right)
\]
\caption{Wide expansion of a P-view}
\label{fig:wide_exp}
\end{figure}
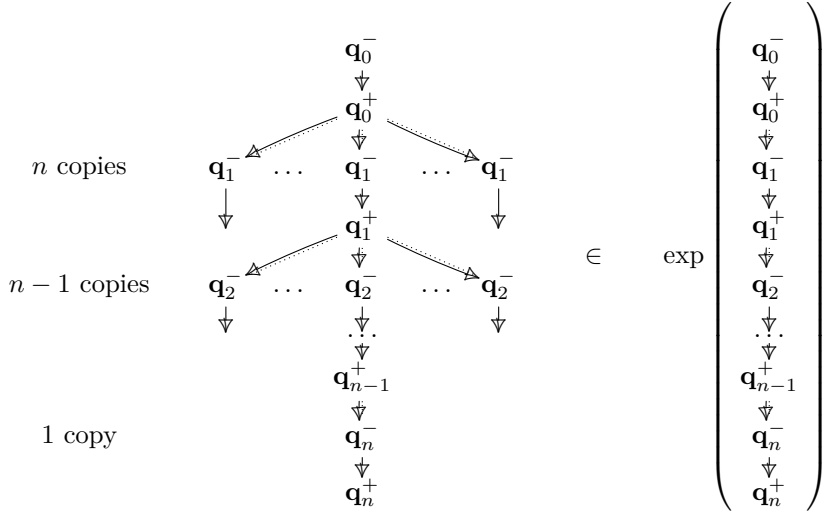

Consider $\p_1$ the sub-augmentation of $\caus{\sigma_1}$ restricted to
prefixes of $s$ -- it is a linear augmentation of length $2n+2$,
as shown on the right hand side of Figure \ref{fig:wide_exp}.
We build the \textbf{wide expansion} $\q_1 \in \exp(\p_1)$ as in
the left hand side of Figure \ref{fig:wide_exp}. It is the unique
$-$-obsessional and $+$-obsessional expansion of $\p_1$ s.t. each
fork of co-depth $2k$ has cardinality $k$ (except for the initial
move). So for any $1\leq k \leq n$, they are $\frac{n!}{(n-k)!}$ copies of $\bq^+_k$.

Now, since $\positions{\sigma_1} = \positions{\sigma_2}$, Proposition
\ref{prop:pos_sig_pos_caus} entails
$\positions{\caus{\sigma_1}} = \positions{\caus{\sigma_2}}$. 
Therefore, there is $\q_2 \in
\exp(\caus{\sigma_2})$ along with some isomorphism $\varphi : \deseq{\q_1} \iso \deseq{\q_2}$.
By abuse of notation, we keep referring to events of $\ev{\q_2}$ with
the same naming convention as in Figure \ref{fig:wide_exp}, this is
justified by the isomorphism $\varphi$.

Now, we study the shape of $\q_2$. It is a tree starting with the
unique initial move $\bq^-_0$, and by courtesy it cannot break the causal
links from positives to negatives; so we may regard it as a tree whose
nodes are the $\bq^+_k$'s. For each $0 \leq k \leq n$, it has exactly
$n!/k!$ nodes of arity $k$ (by \emph{arity}, we mean the number of
children in the tree structure) and by hypothesis its depth is
bounded by $n+1$. 
The essence of the situation is captured by the
following simplified setting:

Fix $n \in \mathbb{N}$. \textbf{Simple trees} are finite trees made
of nodes $\circled{k}$ of arity $k$ for $0\leq k \leq n$. We set
$T_0 = \circled{0}$, and for $k>0$, $T_{k}$ is the tree with root
$\circled{k}$ and $k$ copies of $T_{k-1}$ as children. 
If $t$ is a simple tree, its \textbf{size} $\sharp t$ is its
number of nodes, and its \textbf{depth} is the maximal number of nodes
reached in a path. For instance, the depth of $T_k$ is $k+1$ and 
\[
\sharp T_k = k! \sum_{i=0}^k \frac{1}{i!}\,.
\]

Now, let us consider the set $\Trees(n)$ of simple trees of depth $\leq
n+1$, and having, for $2 \leq k \leq n$, $\frac{n!}{k!}$ nodes
$\circled{k}$, and arbitrarily many nodes $\circled{1}$ and
$\circled{0}$. We prove:

\begin{lemma}\label{lem:trees}
Let $t \in \Trees(n)$ of maximal size. Then, $t = T_n$.
\end{lemma}
\begin{proof}
Seeking a contradiction, assume $t$ is distinct from $T_n$. Consider a
minimal node where they differ, \emph{i.e.} closest to the root -- say
$t$ has some $\circled{p}$ at the row corresponding to $\circled{k}$'s
in $T_n$. If $k=0$ then $p>0$ and this contradicts that the depth of $t$
is less than $n$. So, $k\geq 1$. If $p>k$, then $p\geq 2$. But by
minimality, $t$ is the same as $T_n$ for all rows closer to the root, so
all $\circled{p}$ for $p>k$ are exhausted. Hence, $p<k$. If $k=1$ and
$p=0$, then we may replace $\circled{p}$ with $T_1$, yielding $t' \in
\Trees(n)$ of size strictly greater than $\sharp t$, contradicting
maximality. Otherwise, $k\geq 2$. Then the number of
nodes $\circled{k}$ is fixed, there are fewer of those on this row as
for $T_n$, and they cannot occur on rows closer to the root. Therefore,
there is an occurrence of $\circled{k}$ strictly deeper in $t$. 

We then perform the transformation as in the diagram:
\[
\includegraphics[scale=.6]{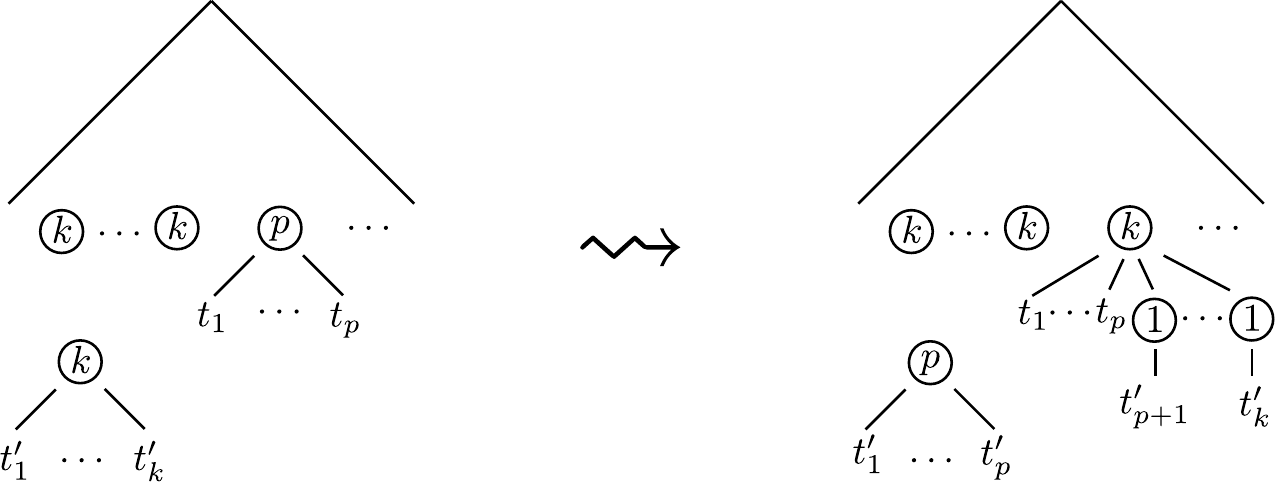}
\]

This yields $t' \in \Trees(n)$. But $\sharp t' > \sharp
t$, contradicting the maximality of $t$. 
\end{proof}

Now, from $\q_2$ we extract a simple tree $t(\q_2) \in \Trees(n)$ as
follows. For each $0 \leq k \leq n$, to each $\bq_{n-k}^+$ we associate
a node $\circled{k}$, with edges as in $\q_2$. Because all P-views in
$\sigma_2$ have length lesser or equal to $2n+2$ and $\q_2 \in
\exp(\caus{\sigma_2})$, $t(\q_2)$ has depth $\leq n+1$. The constraints
on the number of each node are ensured by the isomorphism $\varphi :
\deseq{\q_1} \iso \deseq{\q_2}$. Therefore $t(\q_2) \in \Trees(n)$, and
by Lemma \ref{lem:trees}, $t(\q_2) = T_n$.

This induces directly an isomorphism $\psi$ between $(\q_1,
\leq_{\q_1})$ and $(\q_2, \leq_{\q_2})$. Note that there is a priori no
reason why $\psi$ and $\varphi$ would coincide, so to conclude, we must
still check that $\psi$ preserves $\imc_{\deseq{\q_1}}$, \emph{i.e.}
justification pointers. Assume $\bq^-_j \imc_{\deseq{\q_1}} \bq^+_i$.
Then, $\bq^+_i$ has arity $n-i$, and $\just(\just(\bq^+_i)) = \bq^+_j$
of arity $n-j$. But then, by construction, it follows that 
for any move $a^+\in \ev{\q_1}$ of arity $n-i$, $\just(\just(a))$ has
arity $n-j$. This is transported by the isomorphism $\varphi$, so this
property also holds for $\q_2$. Now, consider $\psi(\bq^+_i) \in
\ev{\q_2}$. Its justifier is some $b^- \in \ev{\q_2}$ such that
$\just(b^-)$ has arity $n-j$. But as arity is preserved by $\psi$, there
is only one move with this property in the causal history of
$\psi(\bq^+_i)$, namely $\psi(\bq^-_j)$. So, $\psi$ preserves pointers.

Finally, $\psi$ also preserves the image in the arena: 
by construction of $\q_1$, all positive moves with the same arity have
the same image, and all negative moves whose justifiers have the same
arity also have the same image. Hence, the image only depends on the
arity, which is a property of $\deseq{\q_1}$; and since $\deseq{\q_1}$
and $\deseq{\q_2}$ are isomorphic, the same holds for $\q_2$. Since
$\psi$ preserves arity and justifiers, it also preserves the image in
the arena. 

We have constructed an isomorphism
$\psi : \q_1 \iso \q_2$. Consider a maximal branch 
\[
\bq^-_0 \imc_{\q_2} \bq^+_0 \imc_{\q_2} \dots \imc_{\q_2} \bq^+_n
\]
of $\q_2$. Since $\q_2 \in \exp(\caus{\sigma_2})$, we have a
morphism $\nu : \q_2 \to \caus{\sigma_2}$. Its image by $\nu$ is
\[
s_1 \imc_{\caus{\sigma_2}} s_1 s_2 \imc_{\caus{\sigma_2}} \dots
\imc_{\caus{\sigma_2}} s_1 \dots s_{2n+2}
\]
a sequence of prefixes of a P-view $s_1 \dots s_{2n+2}$, which, using that
$\q_1$ is an expansion of $\p_1$ and $\psi : \q_1 \iso \q_2$, is
immediately seen to be exactly $s$. Hence, $s \in \sigma_2$ as claimed.

\end{document}